\documentclass[journal,onecolumn, 12pt, draftclsnofoot]{IEEEtran}
\IEEEoverridecommandlockouts
\usepackage{graphicx}
\usepackage{amsmath}
\usepackage{latexsym}
\usepackage{subfigure}
\usepackage{csquotes}
\usepackage{fancyhdr}
\usepackage{epstopdf}
\usepackage{url}
\usepackage{array,color}
\usepackage{algorithm}
\usepackage{algpseudocode}
\usepackage{epsfig}
\usepackage{amssymb}
\usepackage{amsthm}
\usepackage{multirow} 
\usepackage{blindtext}
\usepackage{caption}

\algnewcommand\algorithmicswitch{\textbf{switch}}
\algnewcommand\algorithmiccase{\textbf{Case}}
\algnewcommand\algorithmicassert{\texttt{assert}}
\algnewcommand\Assert[1]{\State \algorithmicassert(#1)}

\algdef{SE}[SWITCH]{Switch}{EndSwitch}[1]{\algorithmicswitch\ #1\ \algorithmicdo}{\algorithmicend\ \algorithmicswitch}%
\algdef{SE}[CASE]{Case}{EndCase}[1]{\algorithmiccase\ #1}{\algorithmicend\ \algorithmiccase}%
\algtext*{EndSwitch}%
\algtext*{EndCase}%

\newcommand{\solution}{MapFirst}
\newcommand{\topoa}{Att}
\newcommand{\topob}{Cernet}

\newcommand{\topoe}{Cogentco}

\newtheorem{mypro}{Proposition}

\newcommand{\qedd}{\hfill $\square$}
\begin{document}
\title{\LARGE{Joint Switch Upgrade and Controller Deployment in \\Hybrid Software-Defined Networks}}
\author{
{\IEEEauthorblockN{Zehua~Guo, Weikun~Chen, Ya-Feng~Liu,~\IEEEmembership{Senior Member,~IEEE,} Yang~Xu,~\IEEEmembership{Member,~IEEE,} Zhi-Li~Zhang,~\IEEEmembership{Fellow,~IEEE}}}
\thanks{ Z. Guo and Z.-L. Zhang were supported in part by NSF grants CNS-1411636, CNS-1618339, CNS-1617729, CNS-1814322 and CNS-1836772. W. Chen and Y.-F. Liu were supported in part by National Natural Science Foundation of China (NSFC) grants 11671419, 11571221, 11688101, and 11631013, and in part by Beijing Natural Science Foundation grant L172020.}
\thanks{Z. Guo and Z.-L. Zhang are with the Department of Computer Science and Engineering, University of Minnesota Twin Cities, Minneapolis, MN 55455 USA (e-mail: guolizihao@hotmail.com, zhzhang@cs.umn.edu).}
\thanks{W. Chen and Y.-F. Liu are with the State Key Laboratory of Scientific and Engineering Computing, Institute of Computational Mathematics and Scientific/Engineering Computing, Academy of Mathematics and Systems Science, Chinese Academy of Sciences, Beijing 100190, China (e-mail: cwk@lsec.cc.ac.cn, yafliu@lsec.cc.ac.cn).}
\thanks{Y. Xu is with the School Of Computer Science, Fudan University, Shanghai, 200433 China (e-mail: xuy@fudan.edu.cn).}
\thanks{Corresponding author: Ya-Feng Liu.}}

\maketitle

\begin{abstract}
To improve traffic management ability, Internet Service Providers (ISPs) are gradually upgrading legacy network devices to programmable devices that support Software-Defined Networking (SDN). The coexistence of legacy and SDN devices gives rise to a hybrid SDN. Existing hybrid SDNs do not consider the potential performance issues introduced by a centralized SDN controller: flow requests processed by a highly loaded controller may experience long tail processing delay; inappropriate multi-controller deployment could increase the propagation delay of flow requests.

In this paper, we propose to jointly consider the deployment of SDN switches and their controllers for hybrid SDNs. We formulate the joint problem as an optimization problem that maximizes the number of flows that can be controlled and managed by the SDN and minimizes the propagation delay of flow requests between SDN controllers and switches under a given upgrade budget constraint. We show this problem is NP-hard. To efficiently solve the problem, we propose some techniques (e.g., strengthening the constraints and adding additional valid inequalities) to accelerate the global optimization solver for solving the problem for small networks and an efficient heuristic algorithm for solving it for large networks. The simulation results from real network topologies illustrate the effectiveness of the proposed techniques and show that our proposed heuristic algorithm uses a small number of controllers to manage a high amount of flows with good performance. 
\end{abstract}


\section{Introduction} 
Software-Defined Networking (SDN) has been widely studied and gradually adapted for campus networks \cite{Mckeown2008OpenFlow}, data center networks \cite{bates2014let}, Wide Area Networks (WANs) \cite{Jain2013B4}, enterprise networks \cite{jin2017magneto}, and Internet exchange points \cite{Gupta2013SDX}. Due to the cost and operational considerations, SDN technology is usually deployed in an incremental fashion. In particular, at each time of network upgrade, only a set of selected legacy network devices (i.e., layer-3 routers and layer-2 switches) are upgraded to SDN switches. AT\&T converted 34\% of its network to SDN by the end of 2016 and virtualized 55\% of its network to software by the end 2017. Its final goal is to reach 75\% softwarization of its network by 2020 \cite{att}. Therefore, the legacy network devices and SDN switches may coexist for a long time. In this paper, we will refer to such a network as a hybrid SDN.

A WAN usually consists of many network devices at geo-distributed locations. A straightforward method to upgrade legacy network devices in WANs to SDN switches is based on the locations of network devices, for example, upgrade a part of WAN with network devices in proximity. The partial upgraded network can enjoy the benefit of SDN, but the performance improvement of the entire network is limited. An efficient solution for network providers is to spread the benefit of upgraded SDNs in the entire network. Based on this consideration, existing studies proposed to upgrade legacy network devices in WANs to SDN switches for different reasons or with different motivations, such as traffic engineering \cite{agarwal2013traffic}\cite{hong2016incremental}, flexible routing \cite{vissicchio2015central}, link failure recovery \cite{chu2015congestion}, power saving \cite{wang2016saving}\cite{jia2018intelligent}, and safe update \cite{vissicchio2017safe}. Essentially, the benefit of SDN is to flexibly control flows. Once a flow traverses an SDN switch, its forwarding path can be flexibly controlled. We call such traffic the programmable traffic. The selection of switches to be upgraded to SDN switches has a great impact on the network's programmable performance. In a WAN, switches tend to have quite different numbers of flows. If we randomly select some of them to upgrade, we may not be able to achieve our objective to maximize the number of programmable flows. Some studies aim to maximize the amount of programmable traffic under the constraint of a given upgrade budget \cite{poularakis2017one}\cite{jia2016incremental}. However, the impact of SDN controller on the network upgrade is not taken into consideration in these works.

The control plane of SDN has evolved from one single controller to multiple controllers to circumvent the limited computational resource of a controller server and avoid single point of failure. For a large wide-area SDN, multiple distributed controllers are physically deployed at different locations to achieve the function of a logically centralized control plane, and the controllers synchronize with each other to guarantee the consistency of the entire network \cite{guo2014improving}\cite{gude2008nox}. To avoid the single-point-of-failure problem, backup controllers maybe deployed \cite{onos}\cite{odl}. If one controller instance crashes, other active instances will still work without service interruption. Popular SDN controllers (e.g., ONOS and OpenDayLight) usually use three controllers to provide resilient service at one location, and the three controllers communicate with each other (e.g., using Raft \cite{ongaro2014search}) to guarantee the network state consistency \cite{ODLcluster}\cite{onoscluster}. 

Deploying multiple controllers in a hybrid SDN should take the following two factors into account. First, controllers should be able to process flow requests from the upgraded SDN switches in a timely manner. If a controller is overloaded, the requests handled by it may suffer from long-tail latency \cite{xie2018cutting}, which might significantly degrade the network performance \cite{ksentini2016using}. Second, the deployment locations of controllers affect the propagation delay of flow requests and network state pulling because the propagation delay in WANs is usually a significant part of the total delay \cite{heller2012controller}\cite{yao2014capacitated}. Existing works \cite{agarwal2013traffic}\cite{hong2016incremental}\cite{poularakis2017one}\cite{jia2016incremental} for the network upgrade did not consider the above factors, which may lead to some undesirable performance degradations. We detail the two factors in Section \ref{motivation}. 

In this paper, we consider the above two factors and propose two objectives: (1) maximizing the number of flows managed by SDN and (2) minimizing the propagation delay of flow requests from SDN switches with an upgrade budget constraint. The first objective aims to control as many flows as possible by upgrading legacy devices to SDN switches, while the second objective aims to reduce the propagation delay between the SDN controllers and switches by effectively deploying a few controllers near SDN switches. Our problem is to find the locations of upgraded switches, the locations of deployed controllers, and the mappings between the controllers and the upgraded switches (i.e., which controllers control which upgraded SDN switches) to maximize the number of controlled programmable flows and at the same time minimize the propagation delay between the controllers and upgraded switches under individual controllers' processing ability and upgrade budget constraints. We first formulate a two-stage optimization problem that optimizes one objective at each stage and then transform the two-stage problem into a one-stage problem to simplify the solution procedure. We prove that with a careful choice of the parameter the optimal solution of the one-stage formulation is also the optimal solution of the two-stage formulation. 

It is worth noting that our work is different from the controller deployment in pure SDNs. In pure SDNs, all switches are SDN switches, and the controllers' deployment only needs to consider two aspects: (1) the controllers' location and number and (2) mapping between SDN switches and controllers. In hybrid SDNs, our problem needs to consider one more aspect: the location and number of the upgraded SDN switches. In our problem, the three aspects are related to each other, and our problem in hybrid SDNs is more complicated than the controllers' deployment in pure SDNs.

To efficiently solve the problem, we analyze the problem's structure and propose several solutions. For the problem in small networks, we accelerate the global optimization solver by proposing a new problem formulation. For the problem in large networks, to further improve the computational efficiency, we propose a heuristic algorithm named \solution, which first orders the (relaxed) mapping variables between controllers and upgraded switches according to their importance/weights (obtained by solving a linear program relaxation) and then sequentially determines the (binary) mapping variables based on their impacts on the objective function. 

We conduct simulations using real network topologies from Topology Zoo \cite{6027859}. The simulation results on multiple topologies verify the effectiveness and efficiency of the new problem formulation. We further compare \solution \ with optimal solutions and other heuristic algorithms, and the results show (1) \solution \ outperforms other heuristic algorithms by using a small number of controllers to facilitate a high amount of flows from upgraded SDN switches, and (2) compared to the optimal solution, \solution \ is able to achieve a comparable performance but enjoys a significantly low complexity. 

The contributions of the paper are summarized as follows:
\begin{enumerate}
\item Novel problem formulation: We identify the impact of the control plane on the network upgrade and formulate an optimization problem to guarantee the performance of the hybrid SDN by jointly upgrading switches and deploying controllers.
\item Efficient solutions: We propose efficient exact and heuristic solutions to solve the problem for both small and large networks. Our simulation results (on real network topologies) demonstrate that our exact solution significantly accelerates the optimization solver, and our proposed \solution \ is able to return a high-quality solution with a significant less CPU time.
\item Theoretical analysis: We provide some analysis on the problem and the solutions. For the problem, we provide a rigorous NP-hardness proof and shed a useful insight that the problem (probably) does not admit constant-ratio approximation algorithms to only approximate the total delay in our problem. In addition, we also identify a special case of the problem which is strongly polynomial time solvable. For the proposed \solution, we analyze its worst-case complexity and prove that it is able to return the global solution of the problem if each controller can control at most one SDN switch.
\end{enumerate}

The rest of this paper is organized as follows: Section \ref{motivation} illustrates the motivation of the paper with some examples. Section \ref{problem} formulates the Joint Switch Upgrade and Controller Deployment (JSUCD) problem, and Section \ref{analysis} analyzes the parameter selection and the complexity of the problem. Section \ref{solution} introduces exact and heuristic solutions for the JSUCD problem. Section \ref{simulation} presents the simulation results and analysis. Section \ref{relatedwork} introduces the related works. Section \ref{conclusion} concludes the paper and presents our future work.

\begin{figure*}[t]
\centering
\subfigure[Network composition]{
\includegraphics[width=1.4in]{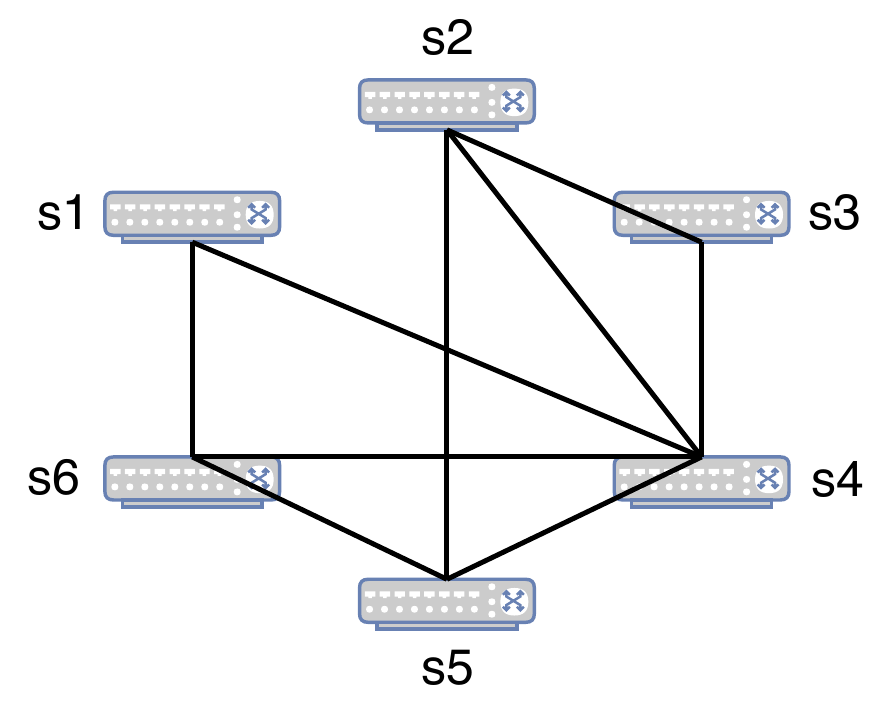}
\label{fig:network}
}
\subfigure[Three switch selections (e.g., red, green, and blue) to maximize the benefit of flexible flow control from SDN]{
\includegraphics[width=1.4in]{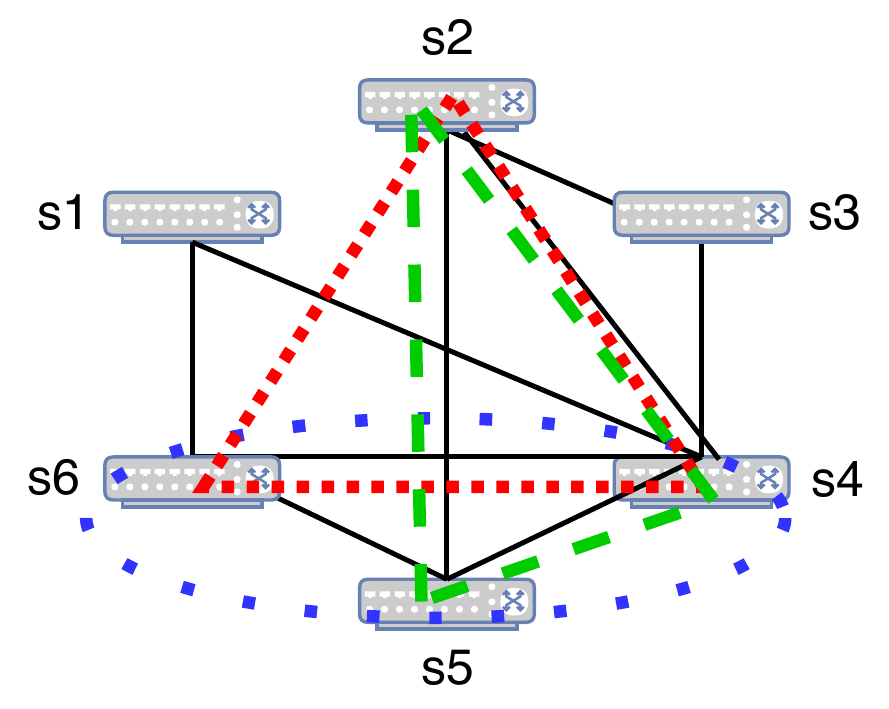}
\label{fig:sw_select}
}
\subfigure[With blue switch selection, deploying one controller cannot satisfy the flow requests from SDN switches]{
\includegraphics[width=1.4in]{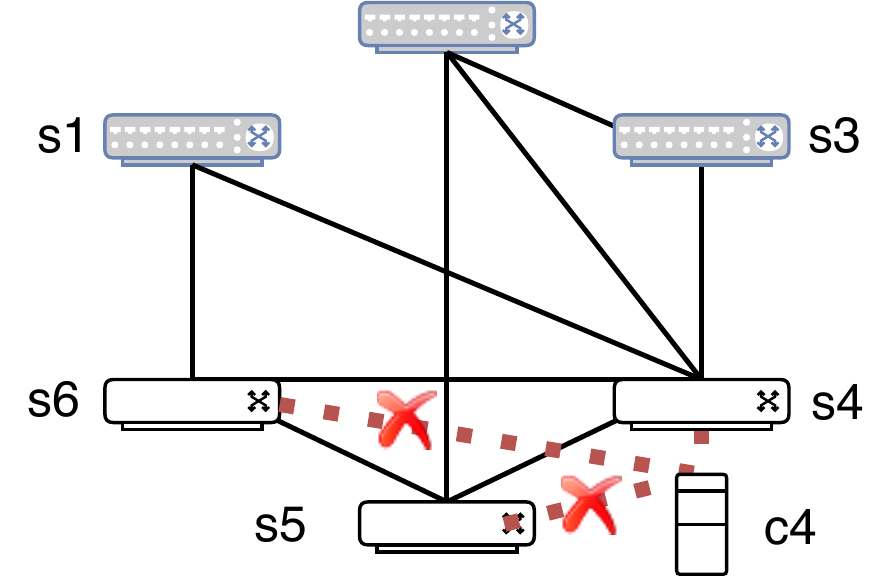}
\label{fig:3s1c}
}
\subfigure[With blue switch selection, deploying three controllers provides much more controller ability than needed]{
\includegraphics[width=1.4in]{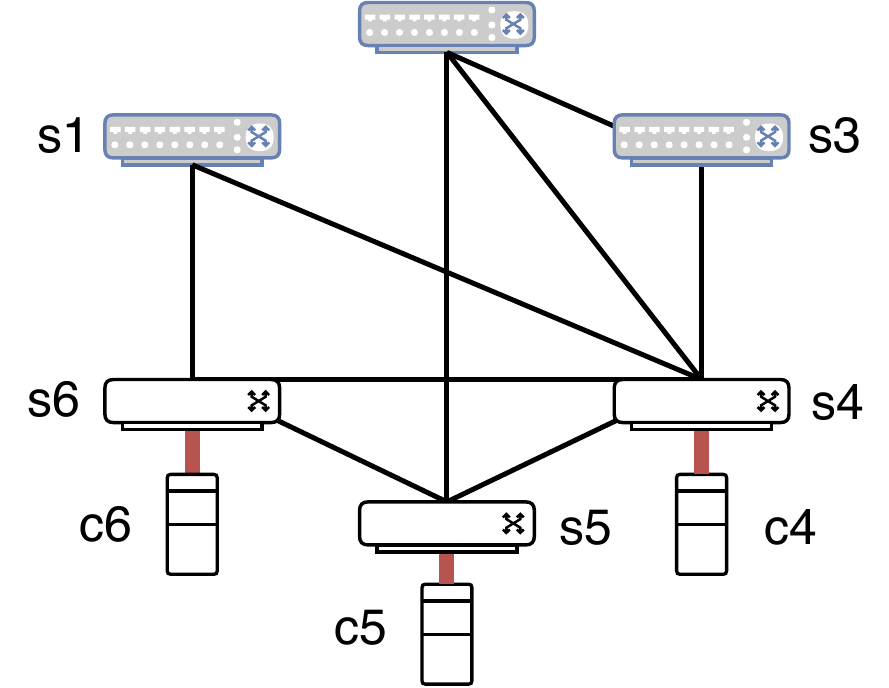}
\label{fig:3s3c}
}
\subfigure[The impact of the controller processing ability on the switch-controller mapping of blue switch selection]{
\includegraphics[width=1.4in]{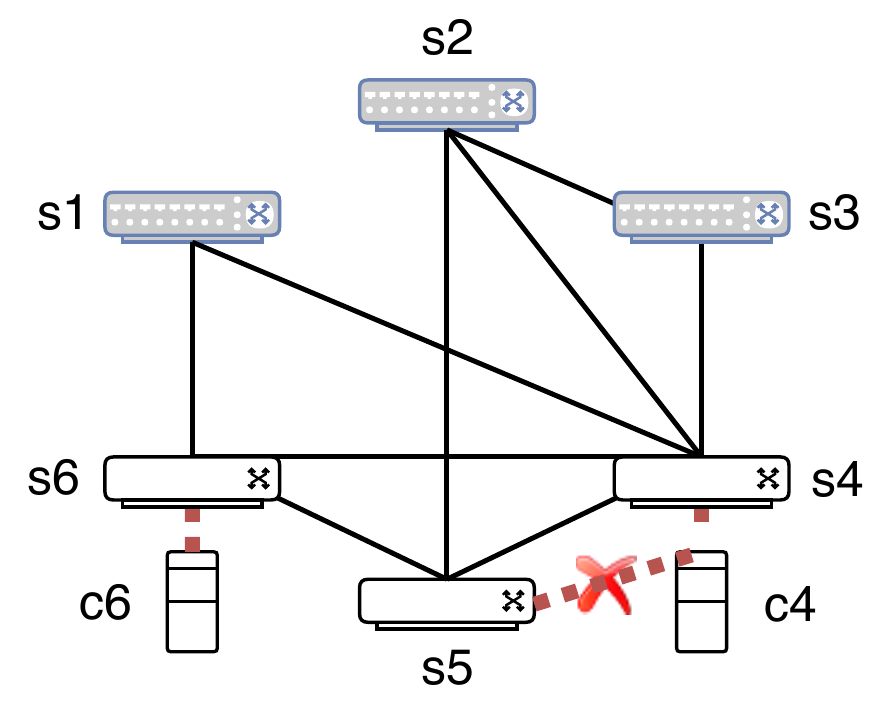}
\label{fig:3s2c_mismapping}
}
\subfigure[A feasible solution of blue switch selection that maximizes the number of programmable flows and minimizes the propagation delay of flow requests with a few controllers]{ 
\includegraphics[width=1.4in]{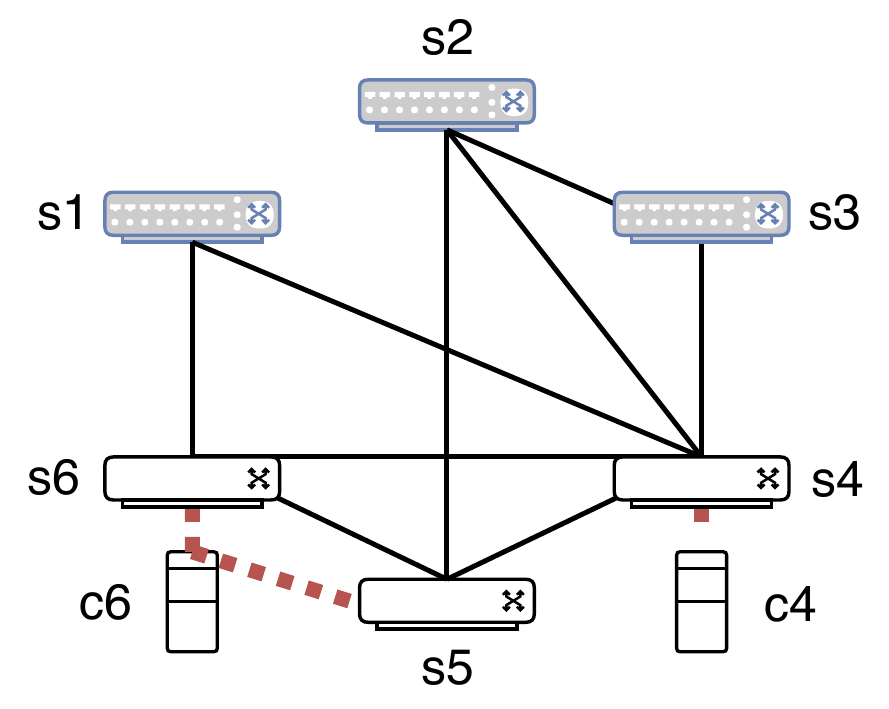}
\label{fig:3s2c_low}
}
\subfigure[A feasible solution of green switch selection that has a higher propagation delay of flow requests than the one with the solution in \ref{fig:3s2c_low}]{
\includegraphics[width=1.4in]{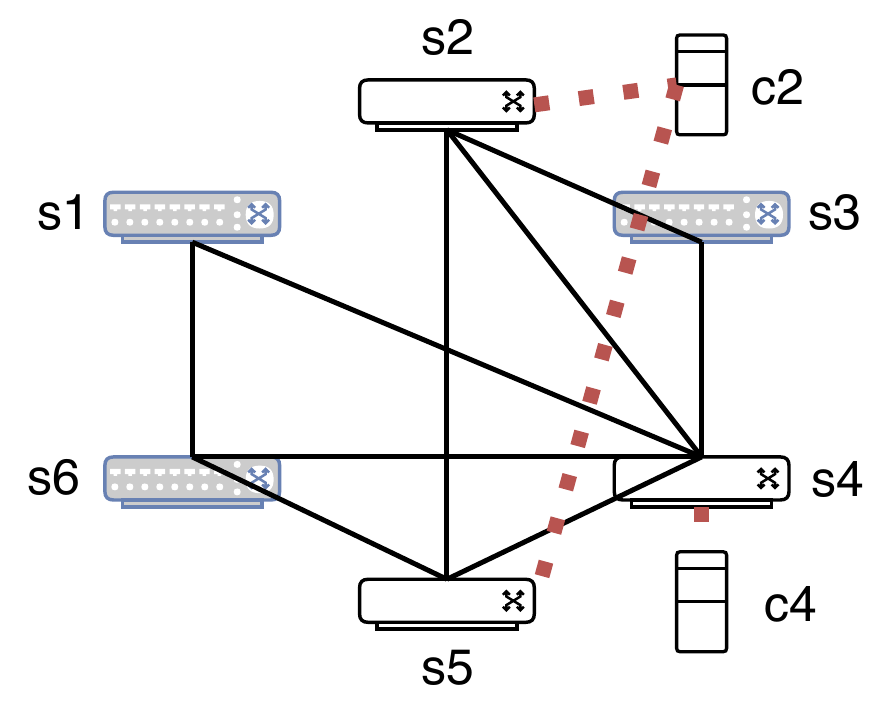}
\label{fig:3s2c_high}
}
\subfigure{
\includegraphics[width=1.2in]{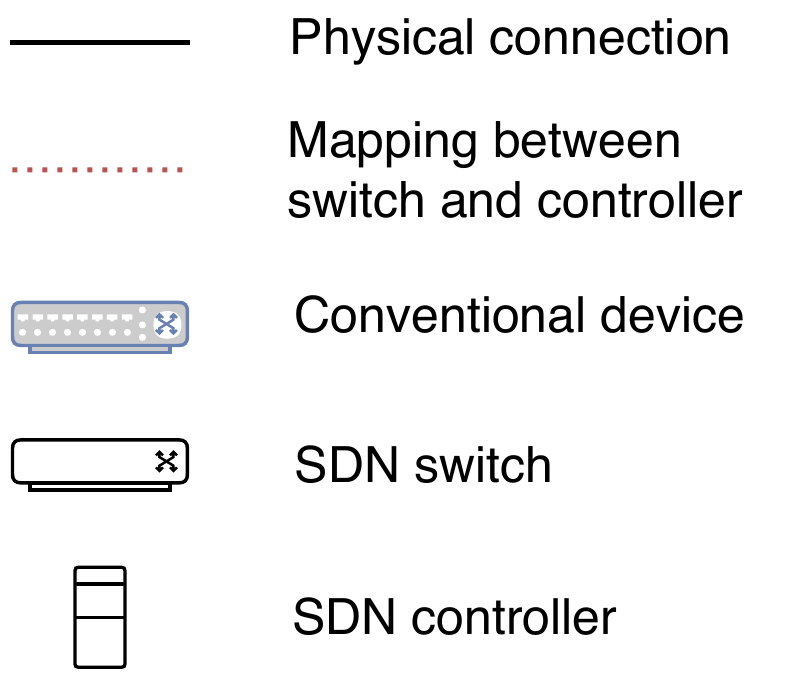}
\label{fig:legend}
}
\caption{A motivation example. }
\label{fig:motivation}
\vspace{-0.4cm}
\end{figure*}

\section{An Example and Motivation}
\label{motivation}
In this section, we use a motivation example in Fig. \ref{fig:motivation} to show how the network upgrade affects the performance of the hybrid SDN.

\subsection{Example background}
Fig. \ref{fig:network} shows the network composition: the network has 6 legacy switches $s_1$ to $s_6$ deployed at six different locations with uneven interconnection. We assume that the number of flows in each switch is proportional to the number of its links, and the total number of flows on a link is normalized as 1. The number of normalized flows at the six switches are: $f_{s1} = 2, f_{s2}= 3, f_{s3} = 2, f_{s4} = 5, f_{s5}= 3, f_{s6} = 3$. The normalized processing ability of each controller is $p_c = 6$. The upgrade budget includes switch upgrade cost and controller deployment cost. In this example, the upgrade budget allows upgrading at most four legacy devices to SDN switches, and the cost of an SDN switch is three times of a controller. In other words, if we upgrade one less switch, we can deploy three more controllers at three locations\footnote{At a location, a controller is physically implemented by a controller instance cluster to prevent single-node-of-failure \cite{onos}\cite{odl}. We can either use two instances or three instances for a cluster. However, the odd number accelerates the primary controller selection in the case of the controller failure, and the production networks usually use three instances \cite{ODLcluster}\cite{onoscluster}. In the rest of the paper, we use a controller to represent a controller cluster at a location since a cluster usually uses one controller to process requests, and the other two are just backup and thus do not process any requests.}. We assume that one SDN switch can only be controlled by exactly one controller, and one controller can potentially control multiple SDN switches. Our problem is to find a feasible solution that contains a set of upgraded switches, a set of controllers, and the mappings between the upgraded switches and controllers to maximize the number of controlled programmable flows and at the same time minimize the delay between the controllers and upgraded switches under the constraints of the upgrade budget and individual controller's processing ability. 

\subsection{Impact of the network upgrade on the hybrid SDN}
\subsubsection{Maximizing the benefit of SDN by selecting upgraded SDN switches}
If a flow traverses one SDN switch, it is a programmable flow and we can enjoy the benefit of SDN by flexibly controlling the flow's forwarding path. Thus, the benefit of SDN depends on the number of programmable flows. We use $A_s$ to denote the total number of programmable flows from the upgraded SDN switches. Our first goal is to maximize $A_s$ with a given switch upgrade budget. There are 20 switch combinations for selecting three switches from six switches to upgrade. Fig. \ref{fig:sw_select} shows three switch selection combinations that control the largest number of flows $A_s = 3+3+5 = 11$: ${red} = \{s_2, s_4, s_6\}$, ${blue} = \{s_4, s_5, s_6\}$, $green = \{s_2, s_4, s_5\}$.

\subsubsection{Guaranteeing the processing performance of flow requests by considering the processing ability of controllers}
We use $blue$ switch selection in the following explanation. The $red$ and $green$ switch selections follow the similar explanation. We use $A_c$ to denote the total processing ability of the deployed controllers. Existing works show when a controller processes more requests than its normal processing load, the processing delay of the requests could be five times longer than the one under the normal load \cite{xie2018cutting} \cite{ksentini2016using}. To maintain the processing performance of flow requests, after a network upgrade, for the entire network, it is necessary to require $A_s \leq A_c$, and for controllers, each one should not process the number of flow requests larger than its normal processing ability. In Fig. \ref{fig:3s1c}, deploying one controller can either control the flows at switch $s_4$ or the flows at the other two switches, and thus deploying one controller is not enough. In Fig. \ref{fig:3s3c}, deploying one controller for each switch can satisfy the control requirement, but the controllers' processing ability of such a deployment is much larger than the number of flows from SDN switches. The extra number of controllers could also bring extra overhead. If we deploy many controllers, the number of switches controlled by a controller would be small, and a controller's control ability on the network would be reduced. Thus, after a network upgrade, we prefer to deploy a few controllers to satisfy the demand of SDN switches. In this example, we only need two controllers since $A_s < A_c = p_c * 2 = 12, f_{s_4} < p_c$, and the load of the other two switches in $red$, $green$, or $blue$ switch selection is less than $p_c$. Thus, a good solution is to upgrade three switches with the number of flows 5, 3, and 3 and deploy only two controllers. 

\subsubsection{Maintaining the good propagation performance of flow requests by considering the locations of SDN switches and controllers}
An SDN switch and a controller interact with each other in two ways: one SDN switch can send a flow request to the controller when it does not know how to process the flow, and a controller can dynamically change the paths of flows on its controlled SDN switches to improve the network performance (e.g., when identifying a congestion). Previous works \cite{heller2012controller}\cite{yao2014capacitated} show that the propagation latency in WANs is the dominant factor among all latencies because the propagation latency bounds the control reactions of a controller that can be executed at a reasonable speed. A long propagation latency could limit availability (e.g., link failure recovery) and convergence time (e.g., network state pulling, routing convergence). Thus, for WANs, an important factor for controller placement is to minimize the total propagation delay among SDN switches and controllers. In WANs, the propagation delay of a request is proportional to the distance between a sender and a receiver. To maintain the good propagation performance of flow requests between the SDN's control plane and data plane, we should minimize the propagation delay of flow requests between SDN switches and controllers, which motivates us to deploy controllers near the upgraded SDN switches.

In the example of $blue$ switch selection, the number of flows in switch $s_4$ is the largest one, and it is very close to the control capacity of a controller. Hence, we select switch $s_4$ to upgrade and deploy controller $c_4$ at location 4 for the switch. The rest two switches are switches $s_5$ and $s_6$, which can be controlled by a controller deployed at location 5 or 6. We use $M_{s_i,c_j}$ to denote the mapping between switch $s_i$ and controller $c_j$. Since the delay between a switch and a controller is minimized when the controller is deployed at the same place of the switch, we have two candidate solutions: $blue_1=\{s_4,s_5,s_6,c_4,c_5,M_{s_4c_4}\}$, $blue_2=\{s_4,s_5,s_6,c_4,c_6,M_{s_4c_4}\}$. 

\subsubsection{Impact of individual controller's processing ability on controller-switch mappings} 
\label{mapping_example}
After selecting upgraded switches and deployed controllers, we should also map each upgraded switch to a controller for control. Each SDN switch can be managed by only one controller, and one controller could possibly manage multiple switches. Fig. \ref{fig:3s2c_mismapping} and \ref{fig:3s2c_low} show two switch-controller mappings of the solution $blue_2$. In the two subfigures, switch $s_4$ is mapped to controller $c_4$, and switch $s_6$ is mapped to controller $c_6$. However, switch $s_5$ has different mappings. In Fig. \ref{fig:3s2c_mismapping}, switch $s_5$ is mapped to controller $c_4$, and the controller's load reaches $3+5 = 8$, which exceeds its maximum processing ability $6$ while the load of controller $c_6$ is only 3. Thus, the mapping between switch $s_5$ and controller $c_4$ is not allowed given that $s_4$ is already mapped to $c_4$. In Fig. \ref{fig:3s2c_low}, switch $s_5$'s mapping is feasible since its mapping to $c_6$ keeps $c_6$'s load at $3+3 = 6$, which does not exceed its maximum processing ability. Thus, switch-controller mappings without considering the controller processing ability constraint could also lead to an infeasible solution.

Fig. \ref{fig:3s2c_high} shows a solution of $green$ switch selection. We use $D_{s_i,s_j}$ to denote the distance (e.g., the length of the direct link or the shortest path) between switches $s_i$ and $s_j$. Because of $D_{s_2,s_5} = 2*D_{s_6,s_5}$, we have $D_{s_5,c_2} = 2*D_{s_6,s_5}$. One can check that optimal solutions of $blue_1$ and $blue_2$ switch selections outperform the optimal solution of $green$ switch selection in terms of the total propagation delay.

\subsection{Design criteria of the network upgrade}
In the above network upgrade, we consider the following factors to upgrade a legacy network to the hybrid SDN:
\begin{enumerate}

\item{maximizing the number of programmable flows from the upgraded SDN switches},
\item{guaranteeing the processing performance of flow requests at controllers}, and
\item{minimizing the total propagation delay of flow requests between controllers and SDN switches}.
\end{enumerate}

The first factor depends on the upgraded switches. The second factor relates to the processing ability of individual controllers. The third factor must take the locations and mappings of controllers and SDN switches into consideration. By considering all these factors, an optimal result would come from a very complicated procedure since the factors couple with each other. For example, to get a better result, mapping selection in Section \ref{mapping_example} could be determined before or simultaneously when selecting switches to upgrade and controllers to deploy. In the next section, we will formulate an optimization problem that takes into account all the above factors for an efficient network upgrade.

\section{Problem formulation}
\label{problem}
In this section, we first introduce constraints and objective functions and then formulate our optimization problem. 

\subsection{System description}
In a network upgrade from a legacy network to a hybrid SDN, we upgrade some of its switches to SDN switches and deploy some controllers to manage all SDN switches. The network consists of $N$ switches deployed at $N$ locations. The number of flows in switch $s_i$ is $R_i \geq 0~(1 \leq i \leq N)$. We use homogeneous controllers, and the processing ability of controller $c_j$ is $A > 0$ $(1 \leq j \leq N)$.

We use $x_{i} = 1$ to denote that switch $s_i$ is upgraded to the SDN switch; otherwise $x_{i} = 0$. Similarly, we use $y_{j} = 1$ to denote that controller $c_j$ is deployed at location $j$ and otherwise $y_{j} = 0$; we use $z_{ij} = 1$ to denote that an SDN switch $s_i$ is mapped to controller $c_j$ and otherwise $z_{ij} = 0$. The default relationships between $x_{i}$, $y_{j}$, and $z_{ij}$ are shown below:
\begin{equation}
\label{switch}
z_{ij}  \leq x_{i},~\forall~i,~\forall~j,
\end{equation}
\begin{equation}
\label{controller}
z_{ij}  \leq y_{j},~\forall~i,~\forall~j.
\end{equation}
Equations \eqref{switch} and \eqref{controller} imply that a feasible switch-controller mapping must satisfy two conditions: a switch is upgraded to an SDN switch, and its mapping controller is deployed.

\subsection{Constraints}
\subsubsection{Controller-switch mapping constraint} 
If switch $s_i$ is upgraded, it must be controlled by only one controller; if switch $s_i$ is not upgraded, it is not controlled by any controller. Thus, we have  
 \begin{equation}
x_{i} = \sum_{j= 1}^{N} z_{ij} ,~\forall~i.
\label{con1}
\end{equation}
If controller $c_j$ is not deployed, it does not control any switches; if controller $c_j$ is deployed, it must control at least one SDN switch. That is:
\begin{equation}
 y_j \leq \sum_{i = 1}^{N} z_{ij},~\forall~j.
 \label{con2}
\end{equation}

\subsubsection{Individual controller's processing ability constraints}
An SDN switch will send a flow request to a controller when it does not know how to process a new flow. The number of flow requests from an SDN switch equals the number of flows traversing the switch, and it should not exceed the controller's processing ability. This can be mathematically written as
\begin{equation*}
\sum_{i=1}^{N}(R_i * x_i* z_{ij}) \leq A,~\forall~j,
\label{old_con3}
\end{equation*}
Since both $x_i$ and $z_{ij}$ are binary variables, we substitute \eqref{switch} into the above nonlinear constraints and reformulate them as the following linear constraints: 
\begin{equation}
\sum_{i=1}^{N}(R_i * z_{ij}) \leq A,~\forall~j.
\label{con3}
\end{equation}
\subsubsection{Upgrade budget constraint} 
The upgrade budget is that the total cost of upgraded switches and deployed controllers is at most $M > 0$. The number of upgraded switches is $\sum_{i = 1}^{N} x_{i} $, and the number of deployed controllers is $\sum_{j = 1}^{N} y_{j} $. Suppose that the cost of an SDN switch is  $\gamma~(\gamma \geq 1)$ times of the cost of a controller. Thus, we have
\begin{equation}
\gamma * \sum_{i = 1}^{N} x_{i} + \sum_{j = 1}^{N} y_{j}  \leq M.
\label{con4}
\end{equation}

\subsection{Objective functions}
There are two objectives in our problem. The first one is to maximize the total number of flows by updating legacy switches to SDN switches:
\begin{equation}
obj_1 = \sum_{i=1}^{N}(R_i * x_{i}).
\end{equation}
The second one is to minimize the total propagation delay of flow requests between SDN switches and controllers by deploying the controllers and mapping the SDN switches to the controllers. We use $D_{ij}~(D_{ij} \geq 0)$ to denote the propagation delay between SDN switch $s_i$ and controller $c_j$, which is proportional to their distance. Thus, we formulate the total propagation delay as follows:
\begin{equation}
obj_2 =  \sum_{i = 1}^{N}\sum_{j= 1}^{N} (D_{ij} * z_{ij}).
\end{equation} 

One main reason for upgrading traditional switches to SDN switches is to enjoy the benefit of programmability to flexibly select the route path. Many existing works adopt the first objective as their objectives \cite{poularakis2017one}\cite{jia2016incremental}. The first objective decides the number and location of SDN switches in a network. The second objective is typically used to optimize the network performance for a network with given SDN switches \cite{heller2012controller}\cite{yao2014capacitated}. In a hybrid SDN, the first objective is usually (much) more important than the second objective.

\subsection{Problem formulation}
The goal of our problem is to maximize the number of flows from the SDN switches and minimize the total propagation delay of flow requests between the SDN switches and controllers by smartly upgrading legacy switches to SDN switches, deploying controllers, and mapping the SDN switches to the controllers. In practice, maximizing the number of programmable flows in our objectives has the first priority, and minimizing the total delay between the SDN switches and controllers has the second priority. Hence, we model the Joint Switch Upgrade and Controller Deployment (JSUCD) problem as a two-stage problem. In the first stage, we maximize the total number of programmable flows without considering the delay objective. This can be done by solving the following problem:

\begin{equation}
\tag{P1}
\label{stage1-a}
\begin{aligned}
& &  F^*= \underset{x,y,z}{\text{max}} \ \ &  \sum_{i=1}^{N} R_{i} x_i  & \\
&  & \text{s.t.}\  \ \ & (\ref{controller})(\ref{con1})(\ref{con2})(\ref{con3})(\ref{con4}),\\
& & & x_i, \, y_j, \, z_{ij} \in \{0,1\},~\forall~i,~\forall~j.      &                 
\end{aligned}
\end{equation}
With the maximum number of programmable flows obtained from problem \eqref{stage1-a}, we minimize the total delay by solving the following problem: 
\begin{equation}
\tag{P2}
\label{stage2-a}
\begin{aligned}
&  & \underset{x,y,z}{\text{min}} \ \ &  \sum_{i=1}^{N} \sum_{j=1}^{N} D_{ij} z_{ij}  & \\
& & \text{s.t.}\ \ \ &  F^* =\sum_{i=1}^{N} R_{i} x_i ,\  (\ref{controller})(\ref{con1})(\ref{con2})(\ref{con3})(\ref{con4}), & \\
& & & x_i, \, y_j, \, z_{ij} \in \{0,1\},~\forall~i,~\forall~j.      &                       
\end{aligned}
\end{equation}

The above two-stage problem formulation needs to solve two problems. An alternative formulation is to combine the two objectives into one objective as follows:
\begin{equation}
\tag{P}
\label{newproblem-a}
\begin{aligned}
& \underset{x,y,z}{\text{max}}
& &  \sum_{i=1}^{N} R_{i} x_i -  \lambda \sum_{i=1}^{N} \sum_{j=1}^{N} D_{ij} z_{ij} & \\
& \text{s.t.} &   & (\ref{controller})(\ref{con1})(\ref{con2})(\ref{con3})(\ref{con4}),\\
& & & x_i, \, y_j, \, z_{ij} \in \{0,1\},~\forall~i,~\forall~j,  &               
\end{aligned}
\end{equation}
where $ \lambda \geq 0 $ is a constant number that gives different weights of the two objective terms. In Section \ref{analysis}, we prove that by choosing the parameter $ \lambda $ appropriately problem \eqref{newproblem-a} is equivalent to the two-stage problem.

\section{Problem analysis}
\label{analysis}
In this section, we prove that under a certain condition, problem \eqref{newproblem-a} is equivalent to the two-stage problem, and analyze the computational complexity of problems \eqref{stage1-a}, \eqref{stage2-a}, and \eqref{newproblem-a}. 

\subsection{Choice of the parameter $\lambda$ in problem \eqref{newproblem-a}} 
Notice that if $D_{ij} = 0$ for all $ i $ and $ j $, problem \eqref{newproblem-a} is equivalent to problem \eqref{stage1-a} for any $\lambda$. In this case, the objective function in problem \eqref{stage2-a} is zero and hence problem \eqref{newproblem-a} is equivalent to the two-stage problem. In the following, we show that even though there are $ i_0 $ and $ j_0 $ such that $ D_{i_0j_0} \neq 0$, by appropriately choosing the parameter $\lambda$, this equivalence still holds. 
\begin{mypro}
	\label{oneproblemcondition}
	Suppose all  $ R_i $ ($ 1 \leq i \leq N $) are integers and there are $ i_0$ and  $j_0$ such that $D_{i_0j_0} \neq 0$. Let $ d $ be the greatest common divisor of $ R_1,\ldots,R_N $, i.e., $ d = \text{max}\{ \bar{d} \ | \ \frac{R_i}{\bar{d}} \in \mathbb{Z}, \ i=1,\ldots,N \} $. If  
	\begin{equation}
	\label{equivalentcondition}
	0 < \lambda < \frac{d}{\sum_{i=1}^{N}\text{max}_{1\leq j\leq N}\{ D_{ij} \}},
	\end{equation}
	then problem \eqref{newproblem-a} is equivalent to the two-stage problem. 
\end{mypro}
\noindent{\textbf{Proof:}} 
Notice that $ \frac{d}{\sum_{i=1}^N\max_{1\leq j\leq N} \{ D_{ij} \}} >0 $ since $ D_{ij} \geq 0$ for all $ i $ and $ j $ and there exist $ i_0$ and  $j_0$ such $ D_{i_0j_0} > 0 $. In the following, we shall use the contradiction argument to show that the solution of problem \eqref{newproblem-a} is also the solution of problems \eqref{stage1-a} and \eqref{stage2-a}. 

Let $ (x^*, y^*, z^*) $ be an optimal solution of problem \eqref{newproblem-a}. Assume that $ (x^*, y^*, z^*) $ is not an optimal solution of problem \eqref{stage1-a}. Then there exists an optimal solution $ (\bar{x}, \bar{y},\bar{z}) $ of problem \eqref{stage1-a} with $  \sum_{i=1}^{N} R_{i} \bar{x}_i > \sum_{i=1}^{N} R_{i} x^*_i $. Because $ R_i$, $\bar{x}_i$, and $x^*_i $ are all integers, by the choice of $ d $,  we have
\begin{equation}
\label{T1}
\sum_{i=1}^{N} R_{i} \bar{x}_i - \sum_{i=1}^{N} R_{i} x^*_i\geq  d.
\end{equation} 
Since $ \bar{x}_i = \sum_{j=1}^N \bar{z}_{ij} \leq 1  $, it follows $ \sum_{j=1}^{N} D_{ij} \bar{z}_{ij} \leq \text{max}_{1\leq j \leq N}\{D_{ij}\}\bar{z}_{ij} $, which, together with the choice of $\lambda$ in \eqref{equivalentcondition}, further implies that
\begin{equation}
\label{T2}
\begin{aligned}
\lambda  \sum_{i=1}^{N} \sum_{j=1}^{N} D_{ij} \bar{z}_{ij} \leq \lambda  \sum_{i=1}^{N}\text{max}_{1 \leq j \leq N}\{ D_{ij}\}\bar{z}_{ij} & \leq  \lambda  \sum_{i=1}^{N}\text{max}_{1 \leq j \leq N}\{ D_{ij}\} < d.&\\
\end{aligned}
\end{equation}
Combining \eqref{T1} with \eqref{T2}, we obtain 
\begin{equation*}
\begin{aligned}
& \sum_{i=1}^{N} R_{i} \bar{x}_i - \lambda \sum_{i=1}^{N} \sum_{j=1}^{N} D_{ij} \bar{z}_{ij} > \sum_{i=1}^{N} R_{i} x^*_i  + d - d 
&= \sum_{i=1}^{N} R_{i} x^*_i \geq  \sum_{i=1}^{N} R_{i} x^*_i - \lambda \sum_{i=1}^{N} \sum_{j=1}^{N} D_{ij} z^*_{ij} , 
\end{aligned}
\end{equation*}
where the last inequality follows from $ D_{ij} \geq 0 $ for all $ i $ and $ j $ and $ \lambda > 0 $. 
However, this contradicts the fact that $ (x^*, y^*, z^*) $ is an optimal solution of  problem \eqref{newproblem-a}. Hence $ (x^*, y^*, z^*) $ is also optimal for problem \eqref{stage1-a}.  Let $ F^* $ be the optimal value of problem \eqref{stage1-a}. Since adding the constraint $ F^* = \sum_{i=1}^m R_i x_i   $
to problem \eqref{newproblem-a} does not change its optimal solution set, then it follows that the optimal solution of problem \eqref{newproblem-a} is also optimal for problem \eqref{stage2-a}. Similarly, one can show that the optimal solution of problem \eqref{stage2-a} is also optimal for problem \eqref{newproblem-a}.
This completes our proof. \qedd

Proposition \ref{oneproblemcondition} shows that one-stage problem \eqref{newproblem-a} with $\lambda$ satisfying the condition in \eqref{equivalentcondition} is equivalent to the two-stage problem, and hence we can solve problem \eqref{newproblem-a} to obtain the desired result instead of solving two problems \eqref{stage1-a} and \eqref{stage2-a}.

\subsection{Complexity analysis}
\label{sec:np}
In this subsection, we analyze the computational complexity of problems \eqref{stage1-a}, \eqref{stage2-a}, and \eqref{newproblem-a}.
\begin{mypro}
	\label{nphard}
	Problems \eqref{stage1-a}, \eqref{stage2-a}, and \eqref{newproblem-a} are all strongly NP-hard.
\end{mypro}
\noindent\textbf{Proof: }	
We only prove the strong NP-hardness of problem \eqref{stage1-a} since the proof of the other two problems is similar. This can be done by establishing a polynomial-time reduction from the 3-partition problem, which is strongly NP-complete \cite{Garey:1979:CIG:578533}\cite{garey1978strong}. Next, we first introduce the 3-partition problem. 

Given a finite set $ \cal S $ of $ 3m $ elements, a bound $ B \in \mathbb{Z}_+ $, and a size $ a_i \in \mathbb{Z}_+ $ for the $ i $-th element with $ \frac{B}{4} < a_i < \frac{B}{2} $ and $ \sum_{i=1}^{3m} a_i = mB  $, can $ \cal S $ be partitioned into $ m $ disjoint sets $ {\cal S}_1,\ldots,{\cal S}_m $ such that $ \sum_{i \in {\cal S}_j} a_i = B $, $ 1\leq j \leq m $?

Given any instance of the 3-partition problem, we construct an instance of problem \eqref{stage1-a} as follows:  
\begin{itemize}
	\item[]
	set $ M = 4m $, $ A=B $, $ \gamma = 1 $, and $ N $ to be any integer
	satisfying  $ N \geq 3m  $; 
	\item[]
	set  $ R_i = a_i $ for $ i \in \{1,\ldots,3m\} $ and $ R_i = 0 $ for $ i \in \{ 3m+1,\ldots,N  \} $.
\end{itemize}
By construction, the objective function in problem \eqref{stage1-a} reduces to 
$\sum_{i=1}^{3m}  a_i x_i$. It is easy to see that, for this constructed instance of problem \eqref{stage1-a}, there exists an optimal solution $ (x,y,z) $ such that $ x_i = 0 $ for $ i \in \{ 3m+1,\ldots,N \} $ and $ z_{ij}= 0 $ for 
$ i \in \{ 3m+1,\ldots,N \} $, $ j \in \{1,\ldots,N \} $. 

Hence the constructed instance of problem \eqref{stage1-a} can be written as 
\begin{equation}
\tag{P1'}
\label{stage1-3par}
\begin{aligned}
& \underset{x,y,z}{\text{max}}
& &  \sum_{i=1}^{3m} a_{i} x_i   \\
& \text{s.t.} &   &     \eqref{controller} \eqref{con1} \eqref{con2}, \\
& & & \sum_{i=1}^{3m} a_i z_{ij} \leq B, \ 1\leq j \leq N, \\
& & & \sum_{i=1}^{3m} x_i  + \sum_{j=1}^{N} y_j \leq 4m,  \\
& & & x_i, \, y_j, \, z_{ij} \in \{0,1\}, \  1\leq i \leq 3m, \ 1\leq j \leq N .                      
\end{aligned}
\end{equation}
In the following, we will show that the answer to the given instance of the 3-partition problem is true if and only if the optimal value of problem \eqref{stage1-3par} is $ mB $. 

Suppose that $ \cal S $ can be partitioned into $ m $ disjoint sets $ {\cal S}_1,\ldots,{\cal S}_m $ such that $ \sum_{i \in {\cal S}_j} a_i = B $ for all $ j\in \{1,\ldots,m\} $. We construct a point $ (\bar{x}, \bar{y}, \bar{z}) $ by setting 
\begin{enumerate}
	\item [] 
	$ \bar{x}_i = 1 $ for $ i \in \{1,\ldots,3m\} $;
	\item []
	$ \bar{y}_j = 1 $ for $ j \in \{1,\ldots,m\} $ and $ \bar{y}_j = 0 $ for $ j \in \{m+1,\ldots,N\} $;
	\item []
	$ \bar{z}_{ij}= 1$ for $ i \in {\cal S}_j $, $ j \in \{1,...,m\} $ and $ \bar{z}_{ij}= 0 $ for the others.
\end{enumerate}

Clearly, $ (\bar{x}, \bar{y}, \bar{z}) $ is a feasible solution of problem \eqref{stage1-3par} with $ \sum_{i=1}^{3m} a_i \bar{x}_i = mB $. As $ \sum_{i=1}^{3m} a_i \bar{x}_i \leq \sum_{i=1}^{3m} a_i = mB  $, we know the optimal value of problem \eqref{stage1-3par} is $ mB $.  

Now suppose that the optimal value of problem \eqref{stage1-3par} is $ mB $ and the corresponding solution is $ (\bar{x}, \bar{y}, \bar{z}) $. Then $ 	\sum_{i=1}^{3m} a_i \bar{x}_i = mB $. This, together with the fact that $ \sum_{i=1}^{3m} a_i = mB $, implies $ \sum_{i=1}^{3m} \bar{x}_i = 3m $.

Then it must follow 
\begin{equation}
\label{firstcondition}
\sum_{j=1}^{N} \bar{y}_j = m. 
\end{equation}
Otherwise, by $ \sum_{i=1}^{3m} \bar{x}_i  + \sum_{j=1}^{N} \bar{y}_j \leq 4m $, we know 
$ \sum_{j=1}^{N} \bar{y}_j < m  $. However, this is impossible since   
\begin{equation*}
\begin{aligned}
\sum_{i=1}^{3m} a_i \bar{x}_i &= \sum_{i=1}^{3m}a_i  \sum_{j=1}^{N} \bar{z}_{ij} =\sum_{j=1}^{N} \sum_{i=1}^{3m}  a_i \bar{z}_{ij} \leq \sum_{j=1}^{N} \bar{y}_j B < mB, &&\\
\end{aligned}
\end{equation*}
where the first equality follows from $ \bar{x}_i = \sum_{j=1}^{N} \bar{z}_{ij} $ in problem \eqref{stage1-3par} and the first inequality follows from
$ \bar{z}_{ij} \leq \bar{y}_j $ and $  \sum_{i=1}^{3m} a_i \bar{z}_{ij} \leq B $ in problem \eqref{stage1-3par}.  By \eqref{firstcondition}, without loss of generality, we  can assume that $ \bar{y}_j =1 $ for $ j \in \{1,\ldots,m\} $ and $ \bar{y}_j=0 $ for $ j \in \{m+1,\ldots,N\} $.
We further show  
\begin{equation}
\label{partition}
\sum_{i=1}^{3m} a_i \bar{z}_{ij} = B, \ 1 \leq j \leq m.  
\end{equation}
Otherwise, we have $ \sum_{i=1}^{3m} a_i \bar{z}_{ij} < B $ for some $ j $. From $ \bar{y}_j=0 $ for $ j \in \{m+1,\ldots,N\} $ and $ \bar{z}_{ij} \leq \bar{y}_j $, we know $ \bar{z}_{ij} = 0$ for $ i \in \{1,\ldots,N\} $, $ j \in \{m+1,\ldots,N\} $. Using this, we have $ 	\sum_{i=1}^{3m} a_i \bar{x}_i = \sum_{j=1}^{N} \sum_{i=1}^{3m}  a_i \bar{z}_{ij} = \sum_{j=1}^{m} \sum_{i=1}^{3m}  a_i \bar{z}_{ij}  < mB $, which contradicts the fact that the optimal value of problem \eqref{stage1-3par} is $ mB $. 
Now, let $ {\cal S}_j = \{ i \, | \, \bar{z}_{ij} = 1, \, 1 \leq i \leq 3m \} $ for $ j \in \{1,\ldots,m\} $. Then it follows from \eqref{partition} that 
\begin{equation*}
\sum_{i \in S_j} a_i = \sum_{i=1}^{3m} a_i \bar{z}_{ij} = B.
\end{equation*}
From $ \bar{x}_i = \sum_{j=1}^N \bar{z}_{ij}$, we have $ {\cal S}_{j_1} \cap {\cal S}_{j_2} = \emptyset $ for any $ j_1 $, $ j_2  \in \{1,\ldots,m\} $ with $ j_1 \neq j_2 $. Furthermore, combining \eqref{partition} with $  \frac{B}{4} < a_i < \frac{B}{2} $, it follows $ \sum_{i=1}^{3m}\bar{z}_{ij} = 3 $, or equivalently $ |{\cal S}_{j}| = 3 $, $   1\leq j \leq m  $. This, together with the fact that $ {\cal S}_1,\ldots,{\cal S}_m $ are disjoint, indicates
\begin{equation*}
{\cal S}_1 \cup  \cdots \cup {\cal S}_m = \{ 1,\ldots,3m \}.
\end{equation*}
Hence the answer to the given instance of the 3-partition problem is true.

Finally, the above transformation can be done in polynomial time. Since the 3-partition problem is strongly NP-complete, we conclude that problem \eqref{stage1-a} is strongly NP-hard. \qedd

Proposition \ref{nphard} implies that, unless P=NP, there are no polynomial time algorithms, which can solve problem \eqref{newproblem-a} to global optimality. Thus, we should develop efficient algorithms for approximately solving the problem especially when the dimension of the problem is large.

\subsection{More insight into intrinsic difficulty of problem \eqref{newproblem-a}}
\label{sec:analysis}
In the above subsection, we basically show that the problem of maximizing the total number of programmable flows is NP-hard. In this subsection, we study another two special cases of problem \eqref{newproblem-a}, i.e., the problem of minimizing the total propagation delay between controllers and SDN switches if all switches can be upgraded to SDN switches and all flows can be programmable, and show that they are still NP-hard. The analysis of the second case shows that there probably does not exist a constant-ratio approximation algorithm for problem \eqref{newproblem-a}. These analysis results provide more insight into the intrinsic difficulty of problem \eqref{newproblem-a}. 

To simplify the following analysis, we replace constraint \eqref{con3} in problem \eqref{newproblem-a} with the following constraint 
\begin{equation}
\sum_{i=1}^{N}(R_i * z_{ij}) \leq A*y_j,~\forall~j. 
\label{tightenedcons3}
\end{equation}
Due to the binary nature of variables $y$ and $z$ and constraint (2), this replacement does not change the solution of the problem.

\subsubsection{Case 1}
In this case, we assume: (i) the processing capacity of the controller is large enough to control all switches' requests, and (ii) the upgrade budget is large enough to upgrade all switches to SDN switches. Mathematically, the above two assumptions can be written as (i) $ A > \sum_{i=1}^N R_i $ and (ii) $ M - \gamma N \geq 1 $. Under these two assumptions, we can reduce problem \eqref{newproblem-a} via the following three steps. First, combining assumption (i) with constraint \eqref{controller}, we can see that {constraint \eqref{tightenedcons3} becomes redundant for problem \eqref{newproblem-a}}, and thus we can remove it.
Second, the above assumptions (i) and (ii) can guarantee $ x_i = 1$ for all $i\in[1, N]$ in the optimal solution of problem \eqref{newproblem-a}. Hence, constraint \eqref{con1} changes into
\begin{equation}
	\label{Rlargecons1}
	\sum_{j=1}^N z_{ij} = 1,~\forall ~i,
\end{equation}
and constraint \eqref{con4} changes into 
\begin{equation}
	\label{Rlargecons4}
	\sum_{j=1}^N y_j \leq M-\gamma N .
\end{equation}
Third, we can remove constraint \eqref{con2} from problem \eqref{newproblem-a}. We argue this as follows: Let us suppose the optimal solution of problem \eqref{newproblem-a} without constraint \eqref{con2} is $ (\bar{x}, \bar{y}, \bar{z}) $ with $ \bar{y}_j > \sum_{i=1} ^N \bar{z}_{ij}$ for some $j \in \mathcal{J} \subseteq [1,N] $. By setting $ \bar{y}_j = 0$ for all $j\in \mathcal{J}$, we obtain a feasible solution for problem \eqref{newproblem-a}, which yields the same objective value as that at $ (\bar{x}, \bar{y}, \bar{z}) $. Combining the above three steps together, we can reduce problem \eqref{newproblem-a} to the problem of selecting appropriate controllers to control all the SDN switches to minimize the total propagation delay as follows:

\begin{equation*}
\label{Rlargeprob}
\begin{aligned}
& \underset{y,z}{\text{min}}
& &  \sum_{i=1}^{N} \sum_{j=1}^{N} D_{ij} z_{ij}  & \\
& \text{s.t.} & & \eqref{controller}\eqref{Rlargecons1}\eqref{Rlargecons4},  &\\
& & & y_{j},\, z_{ij} \in \{0,1\},~\forall~i,~\forall~j.                    
\end{aligned}
\end{equation*}
The above problem is called \emph{p-Median Problem} and is NP-hard \cite{kariv1979algorithmic}. Thus, problem \eqref{newproblem-a} in this special case is also NP-hard. 

\subsubsection{Case 2}
In the case, we assume: (i) all switches have the same number of flows and (ii) the upgrade budget is large enough to upgrade all switches to SDN switches and to deploy controllers to control all upgraded switches. These two assumptions can be written as (i) $ R_1 = \cdots = R_N \triangleq R $ and (ii) $(M-\gamma N) \lfloor \frac{A}{R}\rfloor \geq N,$ where $\lfloor\cdot\rfloor$ is the floor operator. Substituting assumption (i) into constraint \eqref{tightenedcons3}, we obtain
\begin{equation} 
\label{Requalcon3}
\sum_{i=1}^N z_{ij} \leq \Bigl\lfloor \frac{A}{R} \Bigr\rfloor y_j, ~\forall~j.
\end{equation}
It is not difficult to argue that assumptions (i) and {(ii)} guarantee that $ x_i = 1$ for all $i \in[1,N] $ in the optimal solution of \eqref{newproblem-a}. Following steps 2 and 3 of case 1,  we can reformulate problem \eqref{newproblem-a} in this special case as follows:
\begin{equation*}
\label{Requalprob}
\begin{aligned}
& \underset{y,z}{\text{min}}
& &  \sum_{i=1}^{N} \sum_{j=1}^{N} D_{ij} z_{ij}  & \\
& \text{s.t.} & &  \eqref{controller}\eqref{Rlargecons1}\eqref{Rlargecons4}\eqref{Requalcon3},  &\\
& & & y_{j},\,z_{ij} \in \{0,1\},~\forall~i,~\forall~j.                      
\end{aligned}
\end{equation*}
The above problem is called \emph{Uniform Capacitated $p$-Median Problem}. It is a NP-hard problem, and existing studies do not find a constant-ratio approximation algorithm for it \cite{li2017uniform}. Our problem \eqref{newproblem-a} is more general than the above problem. Based on the above analysis, we can conclude that there probably does not exist a good approximation algorithm for our {problem} \eqref{newproblem-a}. 

\section{Problem solution}
\label{solution}
In this section, we propose an exact solution for solving the problem of small networks and an efficient heuristic algorithm for solving the problem of large networks.

\subsection{Exact solution}
\label{sec:final_optimal}
Typically, we can use existing integer program solvers to obtain problem \eqref{newproblem-a}'s optimal solution. Here we use GUROBI solver \cite{gurobi} for solving our problem. GUROBI uses a branch-and-cut \cite{integerprograms} framework and is recognized as one of the fastest integer program solvers \cite{gurobibest}. The branch-and-cut framework usually uses Linear Programming (LP) relaxation to obtain an upper bound. However, the LP relaxation is usually very weak \cite{magnanti1995modeling}, and thus it is difficult for the solver to quickly solve the integer programming problem. In our experiments, we also observe the weakness of the LP relaxation of problem \eqref{newproblem-a}. To accelerate the solution process, we propose a better (re)formulation by strengthening some constraints and adding some valid inequalities by exploiting problem \eqref{newproblem-a}'s structure.

\subsubsection{Strengthening constraints}
First, we strengthen constraint \eqref{con3} as \eqref{tightenedcons3}. Compared to constraint \eqref{con3}, in \eqref{tightenedcons3} we use a tighter upper bound $A*y_j$, which improves the objective value of the LP relaxation of problem \eqref{newproblem-a} and thus reduces the solution time of the branch-and-cut algorithm in GUROBI. The details can be found in \cite{integerprograms}. 

\subsubsection{Adding valid inequalities} Generating efficient cutting planes is a key step in the branch-and-cut framework. In the mixed integer problems, one novel technique is to aggregate multiple constraints together to generate redundant but efficient constraints for the problem. This is because, in the mixed integer problems, some redundant constraints could be used as base constraints to generate cuts and may accelerate the solution process of the branch-and-cut framework \cite{achterberg2016presolve}. Modern solvers can generate cuts automatically but usually do not consider the problem's structure. By exploiting the structure of problem \eqref{newproblem-a}, below we add some aggregated constraints to problem \eqref{newproblem-a} to help the solver efficiently generate cuts. 

Adding all the constraints in \eqref{tightenedcons3} and using the constraints \eqref{con1}, we obtain 
\begin{equation}
\sum_{i=1}^{N}(R_i * x_i) \leq A*\sum_{j=1}^N y_j. 
\label{medialcons2}
\end{equation}
Combining \eqref{medialcons2} with \eqref{con4}, we obtain the following inequality 
\begin{equation}
\sum_{i=1}^{N}(R_i+\gamma A) * x_i \leq A*M.
\label{medialcons3}
\end{equation}
The inequalities \eqref{medialcons2} and \eqref{medialcons3} are redundant for the LP relaxation problem of \eqref{newproblem-a}, but they can be used as base constraints to help the solver to find some knapsack cuts \cite{crowder1983solving} and further accelerate the solution process. The similar technique has been used for the problem of single-source capacitated facility location in \cite{gadegaard2018improved}. 

Substituting \eqref{con1} into the objective of problem \eqref{newproblem-a}, we have
$$  obj= \sum_{i=1}^{N} R_{i} \sum_{i=1}^{N}z_{ij} -  \lambda \sum_{i=1}^{N} \sum_{j=1}^{N} D_{ij} z_{ij} = \sum_{i=1}^{N}  \sum_{i=1}^{N}(R_{i}-\lambda D_{ij} )z_{ij}. $$ Based on the above analysis, we reformulate our problem as the {final problem:}
\begin{equation}
\tag{P'}
\label{finalproblem}
\begin{aligned}
& \underset{x,y,z}{\text{max}}
& &  \sum_{i=1}^{N} \sum_{j=1}^{N} \omega_{ij} z_{ij}  & \\
& \text{s.t.} &   &  \eqref{controller} \eqref{con1}  \eqref{con2} \eqref{con4}\eqref{tightenedcons3}   \eqref{medialcons2}\eqref{medialcons3},\\ 
& & & x_i, \, y_j, \, z_{ij} \in \{0,1\},~\forall~i,~\forall~j,                      
\end{aligned}
\end{equation}
where $x_i$, $y_j$, $z_{ij}$ are design variables and 
\begin{equation}
	\label{omegadef}
		\omega_{ij} = R_i - \lambda D_{ij}.
\end{equation}
We will illustrate the effectiveness and efficiency of the new formulation with simulation results in Section \ref{sec:effect_test}.

\begin{table}[!t]
\centering
\caption{Notations}
\label{table:notation}
\begin{tabular}{|c|l|}
\hline
Notation     & Meaning                                                                                                                                        \\ \hline
$N$           & the number of switches                                                                                                                         \\ \hline
$M$           & the upgrade budget                                                                                                                               \\ \hline
$A$           & the processing capacity of a controller                                                                                                         \\ \hline
$i$           & the index of a switch, $i \in [1,N]$                                                                                                         \\ \hline
$j$           & the index of a controller, $j \in [1,N]$                                                                                                        \\ \hline
$\gamma$           &the cost ratio of an SDN switch to a controller                                                                                                         \\ \hline
$\mathcal{W}$ & the set of weights, $\mathcal{W} = \{\{w_{11},...,w_{1j},...,w_{1N}\},...,\{w_{i1},...,w_{ij},...,w_{iN}\}, ..., \{w_{N1},...,w_{Nj},...,w_{NN}\}, i,j \in [1,N]\}$                                     \\ \hline
$\mathcal{R}$ & the set of the number of flows in each switch, $\mathcal{R} = \{R_i, i \in [1,N]\}$                                                  \\ \hline
$\mathcal{X}$ & the set of updated switches, $\mathcal{X} = \{i\in[1, N] \ |\ x_i =1\}$                                   \\ \hline
$\mathcal{Y}$ & the set of deployed controllers, $\mathcal{Y} = \{j \in [1,N]\ |\ y_j=1 \}$                                             \\ \hline
$\mathcal{Z}$ & the set of the mapping relationship between upgraded switches and controllers, $\mathcal{Z} = \{(i,j)\in [1, N] \times [1, N] \ |\ z_{ij} =1\}$ \\ \hline
\end{tabular}
\vspace{-0.2cm}
\end{table}

\subsection{Heuristic algorithm}
\label{\solution}
The new formulation helps the optimization solver to accelerate the solution process in small networks. However, it still requires a very long time or sometimes is impossible for the solver to find a feasible solution for the problem of large networks. In this section, we propose a heuristic algorithm for solving the problem to achieve the tradeoff between the performance and the time complexity. The heuristic algorithm is based on formulation \eqref{finalproblem}.

The intrinsic difficulty of our problem lies in the interrelationship of the selection of switches to upgrade, the selection of controllers to deploy, and the switch-controller mappings. Due to the interrelationship complexity of the three variables, we cannot change them at the same time. Based on our previous analysis, the mapping variables are more crucial (than switch update and controller deployment variables) since one mapping variable potentially can determine the other two variables. In the following part, we propose the \solution \ algorithm that determines the variables in the order of mappings, switches, and controllers.

\begin{algorithm}[!t]
\caption{\solution \ algorithm}
\label{alg:lp}
{\bf{Input:}} $N$, $M$, $A$, $\gamma$, $\mathcal{W}$, $\mathcal{R}$;

{\bf{Output:}} $\mathcal{X}$, $\mathcal{Y}$, $\mathcal{Z}$;
\begin{algorithmic}[1]
\State $\mathcal{X}=\emptyset$, $\mathcal{Y}=\emptyset$, $\mathcal{Z}=\emptyset$;
\State Generate vector $\mathcal{Z}^{{vec}} = \{{z}^{{vec}}_{l}, l \in [1,N*N] \}$ by solving the LP relaxation of problem \eqref{finalproblem} and sorting the results in the descending order;
\For {${z}^{{vec}}_{l} \in \mathcal{Z}^{{vec}}$}
	\State find switch index $i_0$ and controller index $j_0$ of ${z}^{{vec}}_{l};$
	\If {${i_0}\in \mathcal{X}$}
		\State	continue;
	\EndIf
	\If {\label{ln:greedy_test_map} $ \mathcal{Z} \cup {(i_0, j_0)}, \mathcal{X} \cup {i_0},$ and $\mathcal{Y} \cup {j_0} $ satisfy the constraints in \eqref{finalproblem}}
	     \If  {\label{ln:select0} ${j_0} \in\mathcal{Y}$ and $|\mathcal{X}| * \gamma+ |\mathcal{Y}|+ \gamma \leq M$}
			\State // controller $c_{j_0}$ is already deployed, just upgrade switch $s_{i_0}$
			\State // map switch $s_{i_0}$ to controller $c_{j_0}$
			\State \label{ln:select0_end}$\mathcal{X} \leftarrow \mathcal{X}\cup {i_0}$, $\mathcal{{Z}} \leftarrow \mathcal{{Z}}\cup {(i_0,j_0)}$;
		\ElsIf {\label{ln:select3}${j_0} \notin\mathcal{Y}$ and $|\mathcal{X}| * \gamma+ |\mathcal{Y}|+\gamma+ 1\leq M$}
			\State // deploy controller $c_{j_0}$ and upgrade switch $s_{i_0}$
			\State // map switch $s_{i_0}$ to controller $c_{j_0}$
			\State \label{ln:select3_end} $\mathcal{Y} \leftarrow \mathcal{Y}\cup {j_0}$, $\mathcal{X} \leftarrow \mathcal{X}\cup {i_0}$, $\mathcal{{Z}} \leftarrow \mathcal{{Z}}\cup {(i_0,j_0)}$;
		\EndIf
        \EndIf \label{ln:greedy_test_map_end}
         \If{\label{ln:greedy_cost_constraint}$M-\gamma<|\mathcal{X}| * \gamma+ |\mathcal{Y}|$} 
		\State break;
        \EndIf \label{ln:greedy_cost_constraint_end}
        \label{ln:greedy_end}
\EndFor 
\State \label{ln:return} return $\mathcal{X}, \mathcal{Y}, \mathcal{Z}$;
\end{algorithmic}
\end{algorithm}

The notations used in the algorithm are listed in Table \ref{table:notation}. The idea of our proposed algorithm, \solution, is to first select a switch-controller mapping in the descending order of their importance/weights and then tests whether building the mapping will satisfy the upgrade budget constraint: if yes, the switch and the controller in the mapping is selected; otherwise, a new mapping is tested. The procedure is terminated until there is no budget to build any mapping. Details of \solution \ are summarized in Algorithm \ref{alg:lp}. In line 1, at the beginning of the algorithm, the sets $\mathcal{X}$, $\mathcal{Y}$, and $\mathcal{Z}$ are set to be empty since no switches are upgraded, no controllers are deployed, and there are no mappings between switches and controllers. In line 2, we generate vector $\mathcal{Z}^{{vec}} = \{{z}^{{vec}}_{l}, l \in [1,N*N] \}$. We first relax the binary variables in problem \eqref{finalproblem} to continuous variables in [0,1], and get the LP relaxation solution $\mathcal{Z}^*$ of problem \eqref{finalproblem}. We then sort the solution $\mathcal{Z}^*$ in the descending order to get vector $\mathcal{Z}^{{vec}}$. The sorting operation enables us to test the mapping variables based on their probabilities. Next, we use our customized rounding technique to find the result by sequentially testing each ${z}^{{vec}}_{l} \in \mathcal{Z}^{{vec}}$. In line 4, we get ${z}^{{vec}}_{l}$'s corresponding switch index $i_0$ and controller index $j_0$. Lines 5-7 guarantee that we do not test a switch if it is already upgraded. In lines \ref{ln:greedy_test_map}-\ref{ln:greedy_test_map_end}, we test the mapping between switch $s_{i_0}$ and controller $c_{j_0}$. If the mapping satisfies the constraints of problem \eqref{finalproblem}, we upgrade the switch and deploy the controller at two specific conditions: (1) in lines \ref{ln:select0}-\ref{ln:select0_end}, controller $c_{j_0}$ is already deployed, and we only upgrade switch $s_{i_0}$ when the remaining upgrade budget is $\gamma$ or more; (2) in lines \ref{ln:select3}-\ref{ln:select3_end}, controller $c_{j_0}$ is not deployed, and we upgrade switch $s_{i_0}$ and deploy controller $c_{j_0}$ when the remaining upgrade budget is $\gamma+1$ or more. If either of the two conditions is satisfied, we set the mapping between switch $s_{i_0}$ and controller $c_{j_0}$. In lines \ref{ln:greedy_cost_constraint}-\ref{ln:greedy_cost_constraint_end}, if the remaining budget is less than $\gamma$, then we cannot upgrade any switch, and the algorithm returns the result and stops.

\begin{figure*}[t]
\centering
\subfigure[Test mapping $M_{s_4,c_4}$ and select it.]{
\includegraphics[width=1.4in]{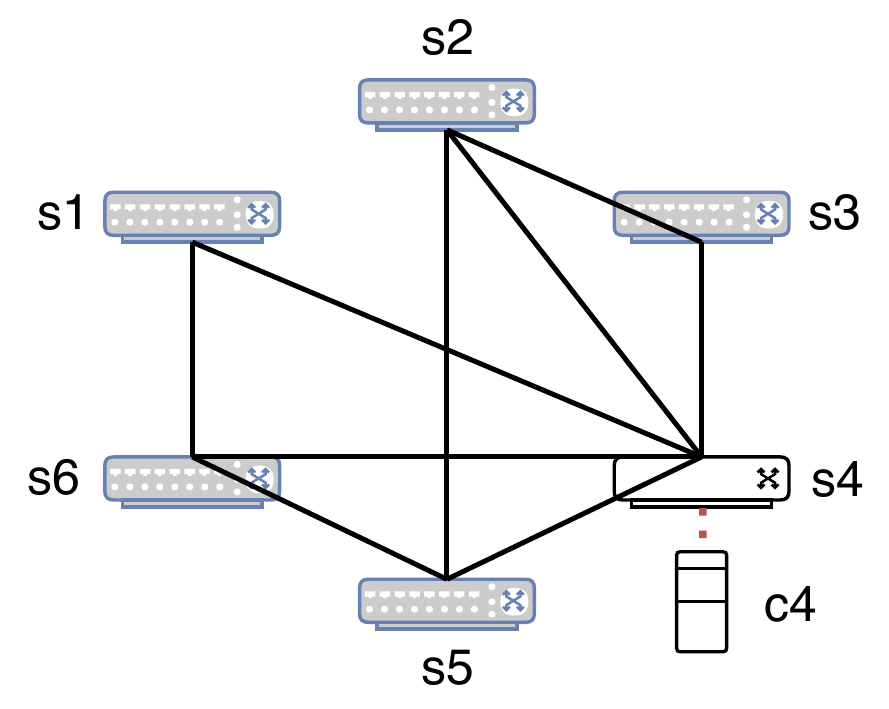}
}
\subfigure[Test mapping $M_{s_5,c_4}$ and not select it.]{
\includegraphics[width=1.4in]{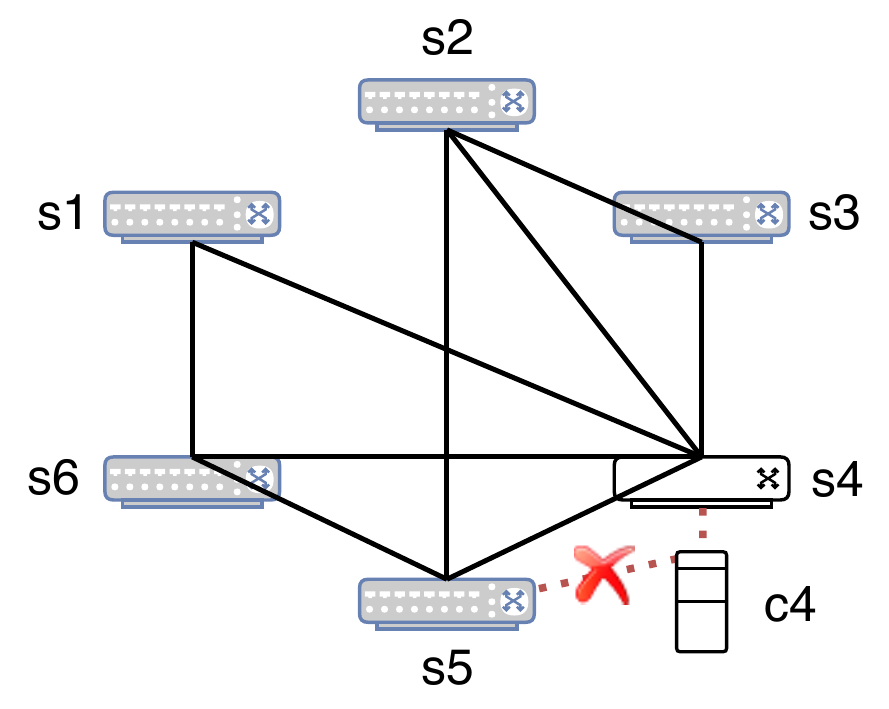}
}
\subfigure[Test mapping $M_{s_6,c_6}$ and select it.]{
\includegraphics[width=1.4in]{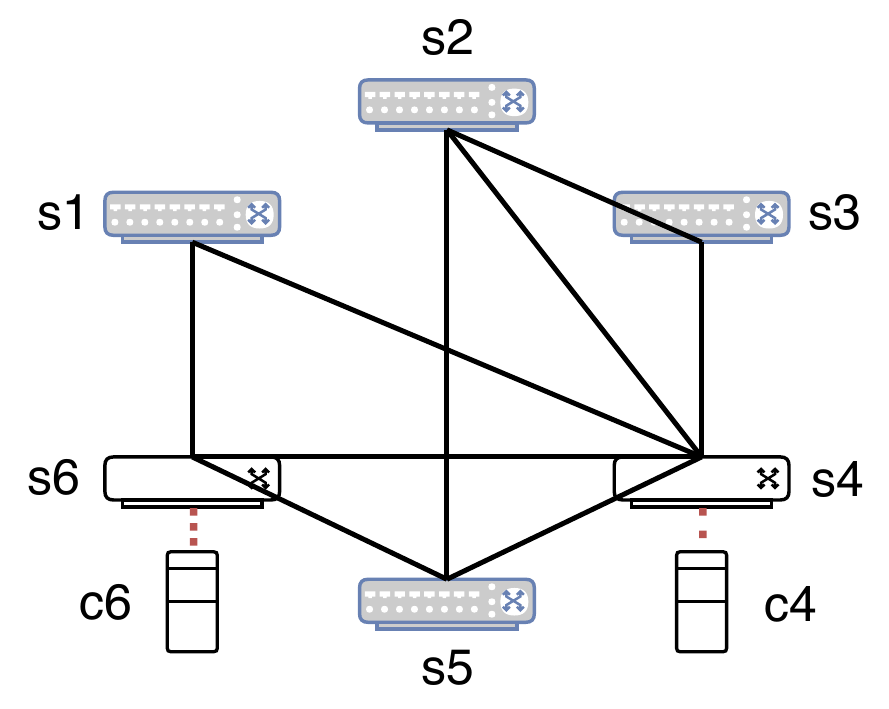}
}
\subfigure[Test mapping $M_{s_5,c_6}$ and select it.]{
\includegraphics[width=1.4in]{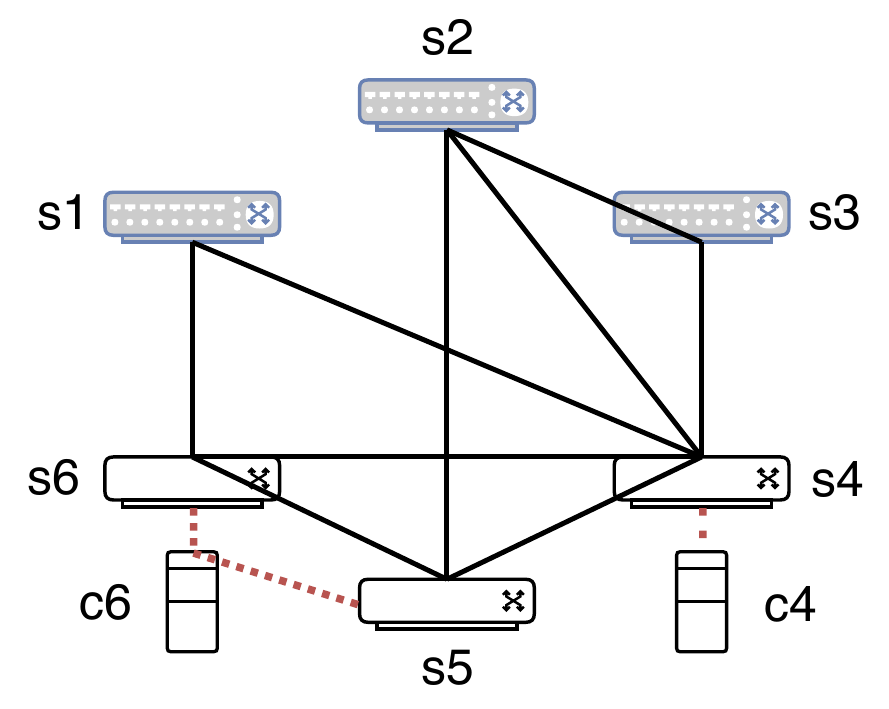}
}
\caption{An example of applying \solution \ to solve the problem. The weights of the mappings satisfy $z_{s_4,c_4}> z_{s_5,c_4}> z_{s_6,c_6}> z_{s_6,c_5}$.}
\label{fig:map_example}
\vspace{-0.4cm}
\end{figure*}
Fig. \ref{fig:map_example} shows an example of applying \solution \ to solve the problem. The mappings are tested in the decreasing order of the mappings' weights. In the figure, the mapping $M_{s_4,c_4}$ is first tested. Since the mapping satisfies the constraints in problem \eqref{finalproblem}, controller $c_4$, switch $s_4$ and the mapping $M_{s_4,c_4}$ are selected. Then, the mapping $M_{s_5,c_4}$ is tested but is not selected because the mapping does not satisfy the processing capacity constraint. Similarly, mappings $M_{s_6,c_6}$ and $M_{s_5,c_6}$ are tested one by one. If building one mapping does not violate the upgrade budget, the related switch, controller, and the mapping itself are selected.

\subsection{Worst-case complexity of \solution}
As shown in Algorithm \ref{alg:lp}, {\solution \ has two main steps: step 1 solves the LP relaxation and gets the weights by sorting the returned solution of the LP relaxation; step 2 generates a binary solution with a customized rounding technique. An LP can be solved in ${\cal O}(n^3*L)$ arithmetic operations by the interior-point methods, where $n$ is the number of variables and $L$ is the length of the input data of the problem \cite{Renegar1988}. Our problem has in total $N^2 + 2N$ variables, and the computational complexity of solving our LP relaxation is ${\cal O}(N^6*L)$. Sorting the returned solution of the LP relaxation takes ${\cal O}(N^2\log N^2) $ operations, and the customized rounding procedure runs at most $N^2$ iterations. In summary, the dominant computational cost of \solution \ is to solve one LP relaxation, and its worst-case complexity is ${\cal O}(N^6*L)$. In sharp contrast, the branch-and-cut framework  in the worst case needs to solve an exponential number of LP relaxations. Therefore, the worst-case complexity of \solution \ is significantly smaller than that of the branch-and-cut framework. 

\subsection{A polynomial time solvable case}
Our analysis in Sections \ref{sec:np} and \ref{sec:analysis} shows that problem \eqref{newproblem-a} is NP-hard, and there probably does not exist a constant-ratio approximation algorithm for it. Therefore, our proposed \solution \ algorithm generally does not have a (constant-ratio approximation) performance guarantee. In this subsection, we consider a special case of problem \eqref{newproblem-a}, where each controller can control at most one SDN switch (i.e., $ R_i + R_j > A$ for all $i \neq j $), and show that MapFirst is guaranteed to find the optimal solution of problem \eqref{newproblem-a} in this special case.

\begin{mypro}
	\label{easycase}
	If $ R_i + R_j >A$ for all $i \neq j $, problem \eqref{newproblem-a} can be solved (to globally optimality) in $ \mathcal{O}(N^3) $, and MapFirst is guaranteed to find the optimal solution of problem \eqref{newproblem-a}.
\end{mypro}
\begin{proof}
	Without loss of generality, we assume $ 0 < R_i \leq A$ for all $i \in [1,N] $. Combining this assumption with the assumption $ R_i + R_j > A$ for all $i \neq j $, constraint \eqref{tightenedcons3} reduces to $ \sum_{i=1}^N z_{ij} \leq y_j $, which, together with \eqref{con2}, implies
	\begin{equation}
	\label{Ronecons}
	y_j = \sum_{i=1}^N z_{ij}, ~\forall ~j. 
	\end{equation}
	By substituting \eqref{con3} and \eqref{Ronecons} into \eqref{con4}, we have
	\begin{equation}
	\label{Ronebudget}
	\sum_{i=1}^N \sum_{j=1}^N z_{ij} \leq \Bigl\lfloor \frac{M}{\gamma+1}\Bigr\rfloor.
	\end{equation}
	From \eqref{Ronecons}, we can simply remove constraint \eqref{controller} from problem \eqref{newproblem-a} and reformulate it as follows:
	\begin{equation}
	\label{Roneprob}
	\begin{aligned}
	& \underset{z}{\text{max}}
	& &  \sum_{i=1}^{N} \sum_{j=1}^{N} \omega_{ij} z_{ij}  & \\
	& \text{s.t.} & & \sum_{j=1}^N z_{ij} \leq 1,~\forall~i,  \ \sum_{i=1}^N z_{ij} \leq 1,
	\ \forall~j,   &\\
	& & & \eqref{Ronebudget},\, z_{ij} \in \{0,1\},~\forall~i,~\forall~j.                      
	\end{aligned}
	\end{equation}
Due to $ R_i >0$ for all $ i \in [1,N]$, the definition of $w_{ij}$ in \eqref{omegadef}, and the selection of $ \lambda $ in \eqref{equivalentcondition}, we know $ \omega_{ij} > 0$ for all $ i,~j \in[1,N]$. Hence, inequality \eqref{Ronebudget} can be rewritten as an equality, and problem \eqref{Roneprob} is a {\it $ k $-cardinality assignment} problem. It has been shown in \cite{dell1997k} that the LP relaxation of problem \eqref{Roneprob} is tight, i.e., solving the linear relaxation of problem \eqref{Roneprob} with an appropriate method (e.g., the simplex method) returns a binary solution. {Furthermore, using the {\emph{primal algorithm}} presented in \cite{dell1997k}, we can solve problem \eqref{Roneprob} with $ \mathcal{O}(N^3) $ complexity}. The proof is completed.
\end{proof} 

The above proposition shows that, if each controller can control at most one SDN switch, after some preprocessing steps, the LP relaxation in \solution \ is tight, and \solution \ can return an optimal solution of problem \eqref{newproblem-a}. 

\section{Simulation results}
\label{simulation}
\subsection{Simulation setup} 
In our simulation, we choose some backbone topologies from Topology Zoo \cite{6027859} to evaluate the performance of our proposed solution. In the topologies, each node has a latitude and a longitude. Since the propagation delay of a flow request is usually proportional to the distance from a switch to a controller, we use the distance between two nodes to represent the propagation delay between an SDN switch and a controller. We use $M\_percent$ to denote the ratio of the given upgrade budget to the cost of upgrading all switches in the network. In our simulation, $M\_percent$ changes from 5\% to 50\%. We follow the assumption in \cite{poularakis2017one}\cite{jia2016incremental} that the number of programmable flows in an SDN switch is proportional to the number of its links. In practice, we can analyze the traffic history at each switch to get the real traffic statistic. We have analyzed more than 50 topologies in Topology Zoo and find that most nodes have two or three links. We set the normalized processing ability of a controller $A = 50$, and thus one controller is able to control at least {ten} SDN switches on average. Note that our problem can take into consideration of heterogeneous controllers by setting the processing abilities of different controllers with different values. The recommended system requirement of one OpenDayLight controller instance is 8 Cores, 8G RAM and 64GB storage \cite{odltest}, and the resilient three-controller deployment requires at least three physical servers and costs about \$2000. One typical SDN switch is about \$8000. Hence, we set the cost ratio between an SDN switch and one controller $\gamma = 4$ . 

\subsection{Compared algorithms}
\begin{enumerate}
\item Optimal: the optimal solution of problem \eqref{finalproblem} that maximizes the number of programmable flows and minimizes the total propagation delay between upgraded SDN switches and controllers. We solve the problem by using GUROBI \cite{gurobi}.
\item FlowOnly: the optimal solution of problem \eqref{stage1-a} that only maximizes the number of programmable flows. Problem \eqref{stage1-a} is also solved by using GUROBI \cite{gurobi}.
\item \solution: we first use GUROBI \cite{gurobi} to solve the LP relaxation of problem \eqref{finalproblem}, and then use a customized rounding technique to sequentially determine the variables in the order of mappings, switches, and controllers. The details can be found in Section \ref{\solution}. 
\item WeightFirst: this algorithm is similar to \solution, but the key difference is that it greedily tests and picks the switch-controller mapping in the descending order of weights $\left\{w_{ij}\right\},i, j\in [1,N]$. 
\end{enumerate}

\begin{table}[t]
	\caption{Computational results of (a) original formulation \eqref{newproblem-a},  and (b) formulation \eqref{finalproblem}.}
	\label{allgapimprovement}
	\centering	
	\begin{tabular}{|c|c|c|c|c|}
		\hline
		\multirow{2}{*}{Topology name} &  \multicolumn{2}{c|}{Topology info}& \multicolumn{2}{c|}{Elapsed CPU time (seconds)}  \\
		\cline{2-5} &{\# of nodes }&{\# of links}& (a)       & (b) \\
		\hline
		Colt			   &153 &185&  2046.46        &33.04  \\
		\hline 
		GtsCe	                &150 &193& 1365.47  &89.17  \\
		\hline 
		Cogentco	          &197 &245&  3600.00         & 419.49       \\
		\hline
		Condensed\_west\_europe & 278 &394 &3600.00          & 576.93\\
		\hline
		Condensed & 463 & 620 & 2721.63  &   1013.31\\
		\hline
	\end{tabular}
	\vspace{-0.4cm}
\end{table}

\subsection{Effectiveness of formulation \eqref{finalproblem}}
\label{sec:effect_test}
We test the effectiveness of new formulation \eqref{finalproblem} under different topologies from Topology Zoo \cite{6027859}. We set a time limit of 3600 seconds for the branch-and-cut algorithm in GUROBI (i.e., we terminate the algorithm if it does not find the solution within 3600 seconds) and set $M\_percent = 50\%$ for each topology. Table \ref{allgapimprovement} summarizes the computational results. We can see from the table that: (1) for the problem instances Cogentco and Condensed\_west\_europe, GUROBI can successfully solve the new formulation \eqref{finalproblem} within the given time limit but fail to solve the original formulation \eqref{newproblem-a}; (2) for the problem instances Colt, GtsCe, and Condensed, GUROBI can successfully solve both problem formulations within the given time limit but solving formulation \eqref{finalproblem} takes significantly less time than solving formulation \eqref{newproblem-a}. These simulation results clearly show the effectiveness of formulation \eqref{finalproblem}, i.e., the newly added constraints \eqref{tightenedcons3}, \eqref{medialcons2}, and \eqref{medialcons3} indeed work and significantly accelerate the solution process.

\begin{figure*}[t]
\centering
\subfigure[\topoa]{
\includegraphics[width=2in]{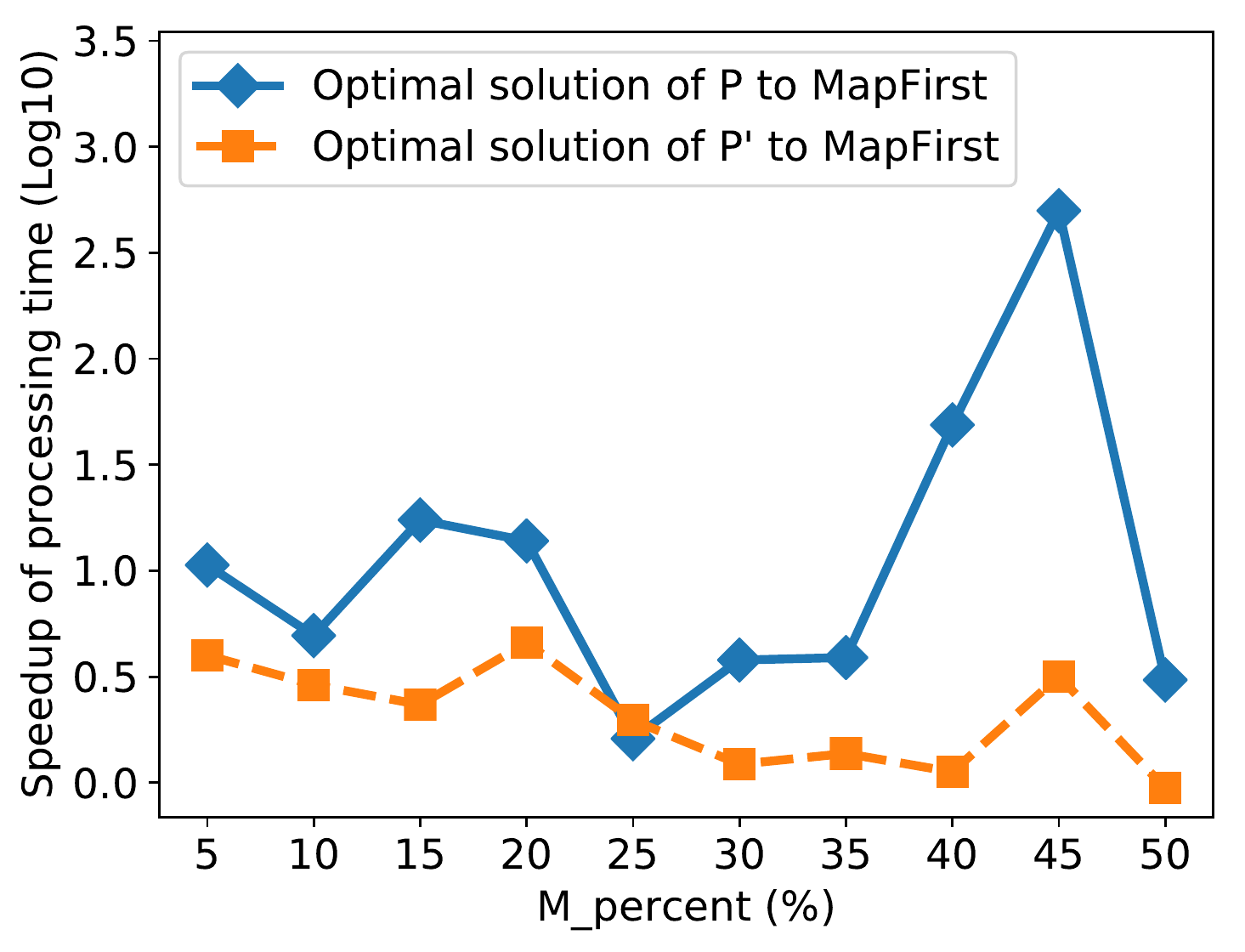}
}
\subfigure[\topob]{
\includegraphics[width=2in]{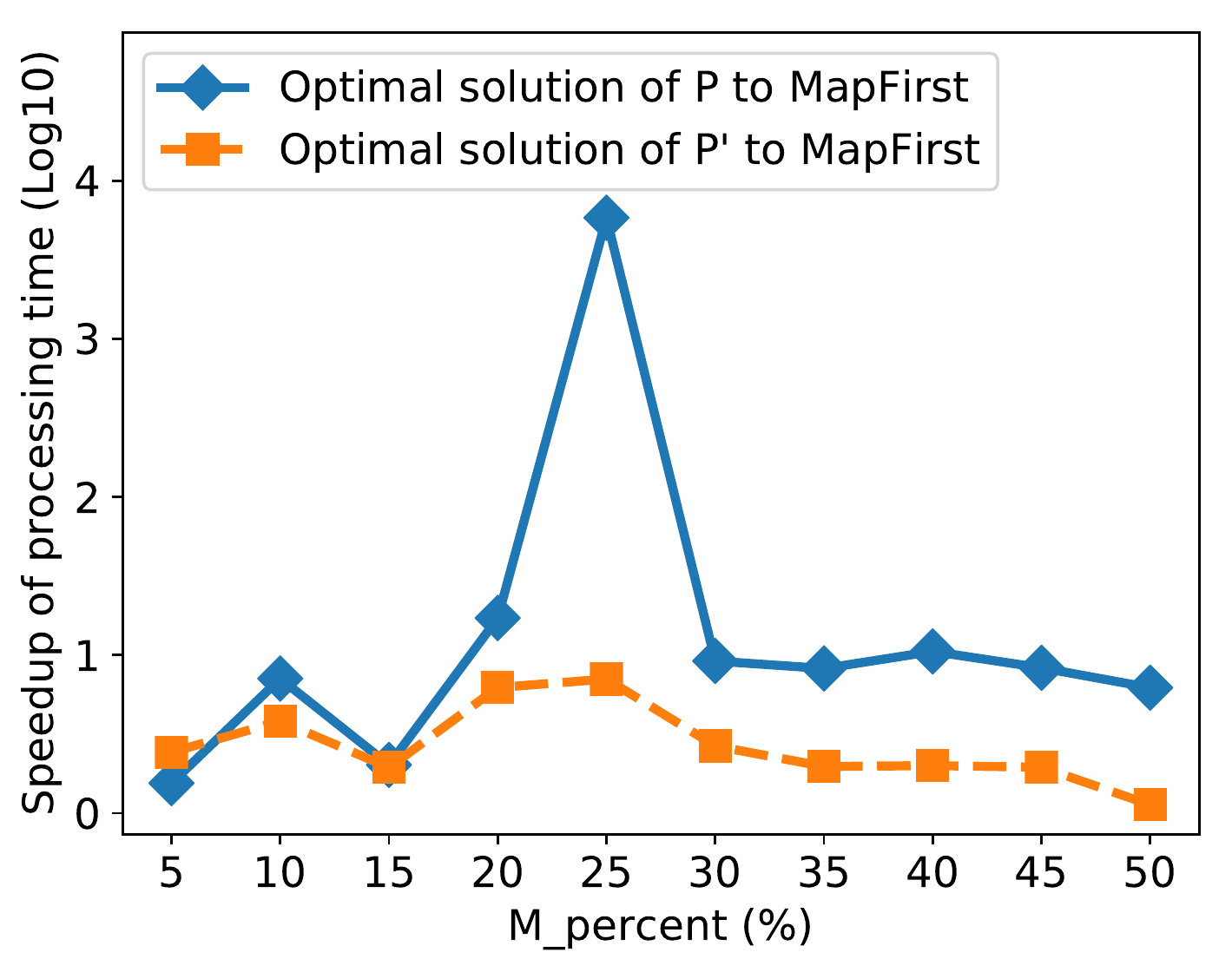}
}
\subfigure[\topoe]{
\includegraphics[width=2in]{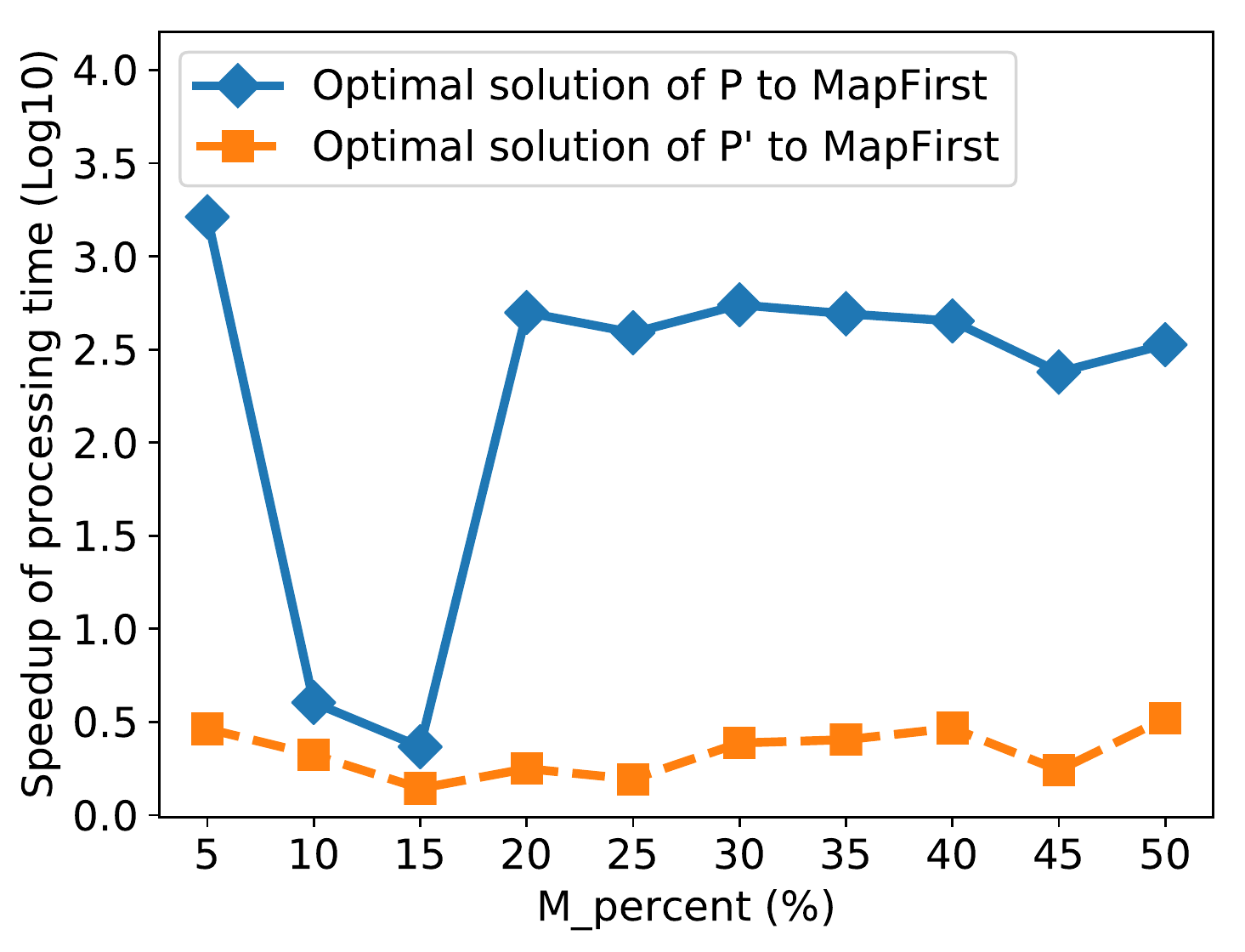}
}
\caption{Speedup of the CPU time in units of $Log{10}$.}
\label{fig:process_time}
\end{figure*}

\subsection{Performance of CPU time }
In the rest of this section, we focus on three topologies \topoa, \topob, and \topoe \ from Topology Zoo \cite{6027859}. More specifically, \topoa \ is a small topology with 25 nodes and 57 links, \topob \ is a medium-size topology with 41 nodes and 59 links, and \topoe \ is a large topology with 197 nodes and 245 links. All of our simulations below are performed on these three topologies. 

We use the ratio of the CPU time of an algorithm to that of \solution \ as the metric to measure the efficiency of the algorithm. Fig. \ref{fig:process_time} shows the results, where y-axis is in the $Log{10}$ scale. Recall that problem \eqref{newproblem-a} is the original problem, and problem \eqref{finalproblem} is problem \eqref{newproblem-a} with strengthened constraints and valid inequalities. We can clearly observe from Fig. \ref{fig:process_time} that \solution \ is the fastest solution in all cases. More specifically, \solution \ is 499 and 5828 times faster than directly using GUROBI to solve problem (P) when M\_percent = 45\% for Att and M\_percent = 25\% for Cernet, respectively. In Fig. \ref{fig:process_time}(c), we set a time limit of 3600 seconds for solving \eqref{newproblem-a}. We only get the results of M\_percent  = 10\% and 20\%, and we use the CPU time limit 3600 seconds as the CPU time of other cases of M\_percent. In other words, GUROBI failed to solve problem \eqref{newproblem-a} directly within 3600 seconds. However, for all cases of M\_percent, GUROBI can successfully solve problem \eqref{finalproblem} within the time limit. In fact, we can see from Fig. \ref{fig:process_time} that it is much more efficient to solve problem \eqref{finalproblem}  than problem \eqref{newproblem-a} in most cases. These simulation results verify that problem \eqref{finalproblem} is indeed a better formulation than problem \eqref{newproblem-a} because added constraints and inequalities in \eqref{finalproblem} are very effective to speed up the GUROBI solver, and our proposed MapFirst is effective. 

\subsection{Performance of programmable flows}

\begin{figure*}[t]
\centering
\subfigure[\topoa]{
\includegraphics[width=2in]{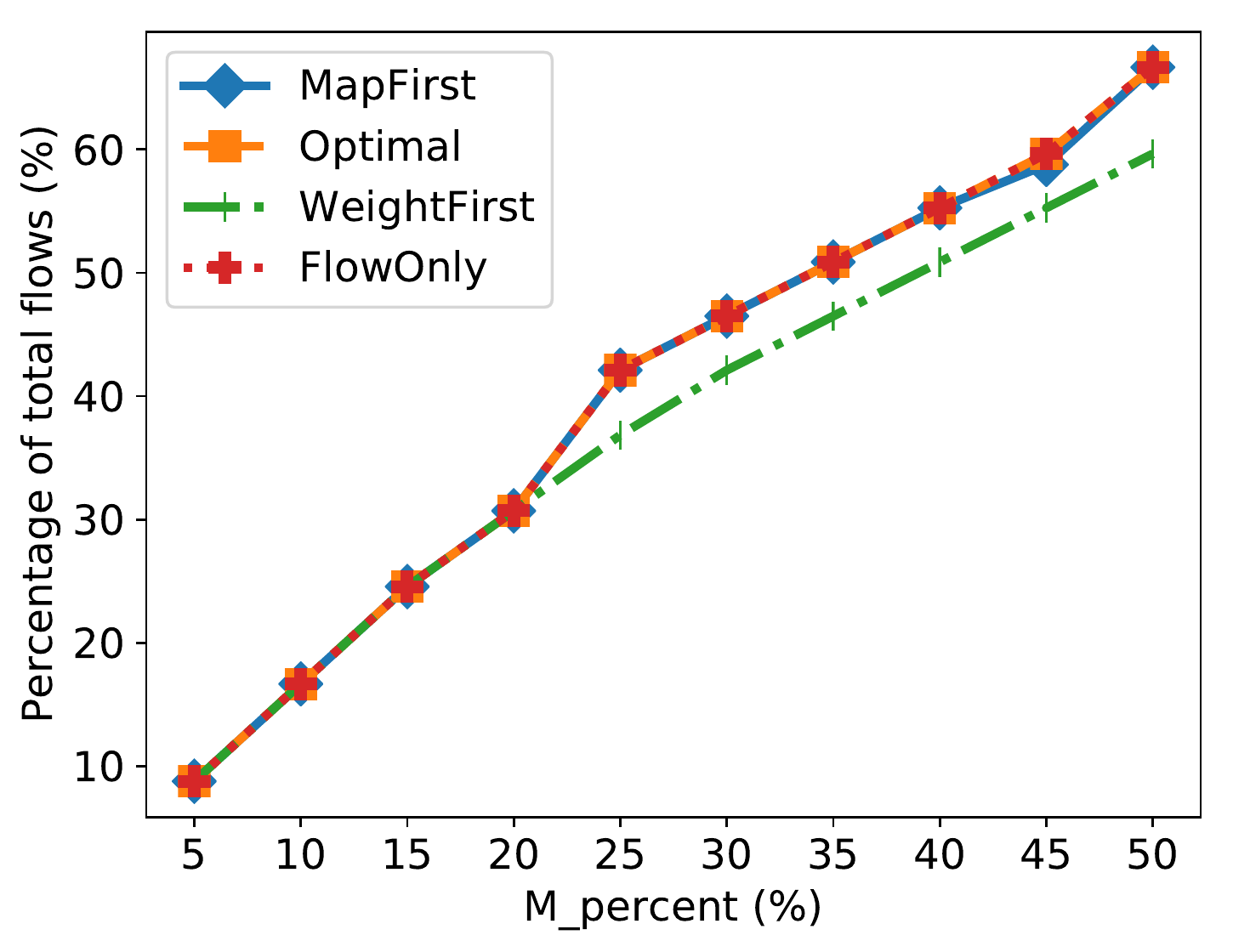}
}
\subfigure[\topob]{
\includegraphics[width=2in]{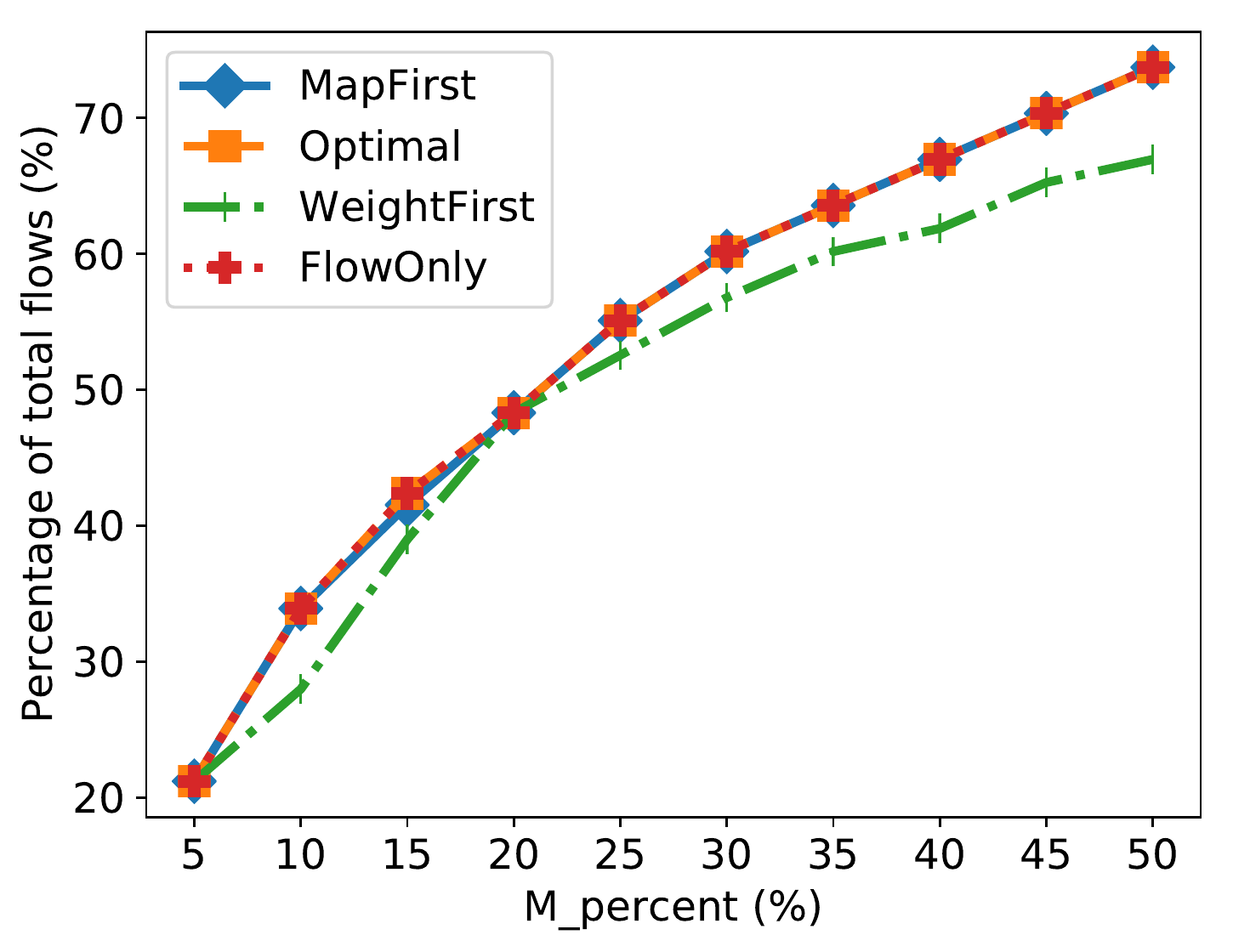}
}
\subfigure[\topoe]{
\includegraphics[width=2in]{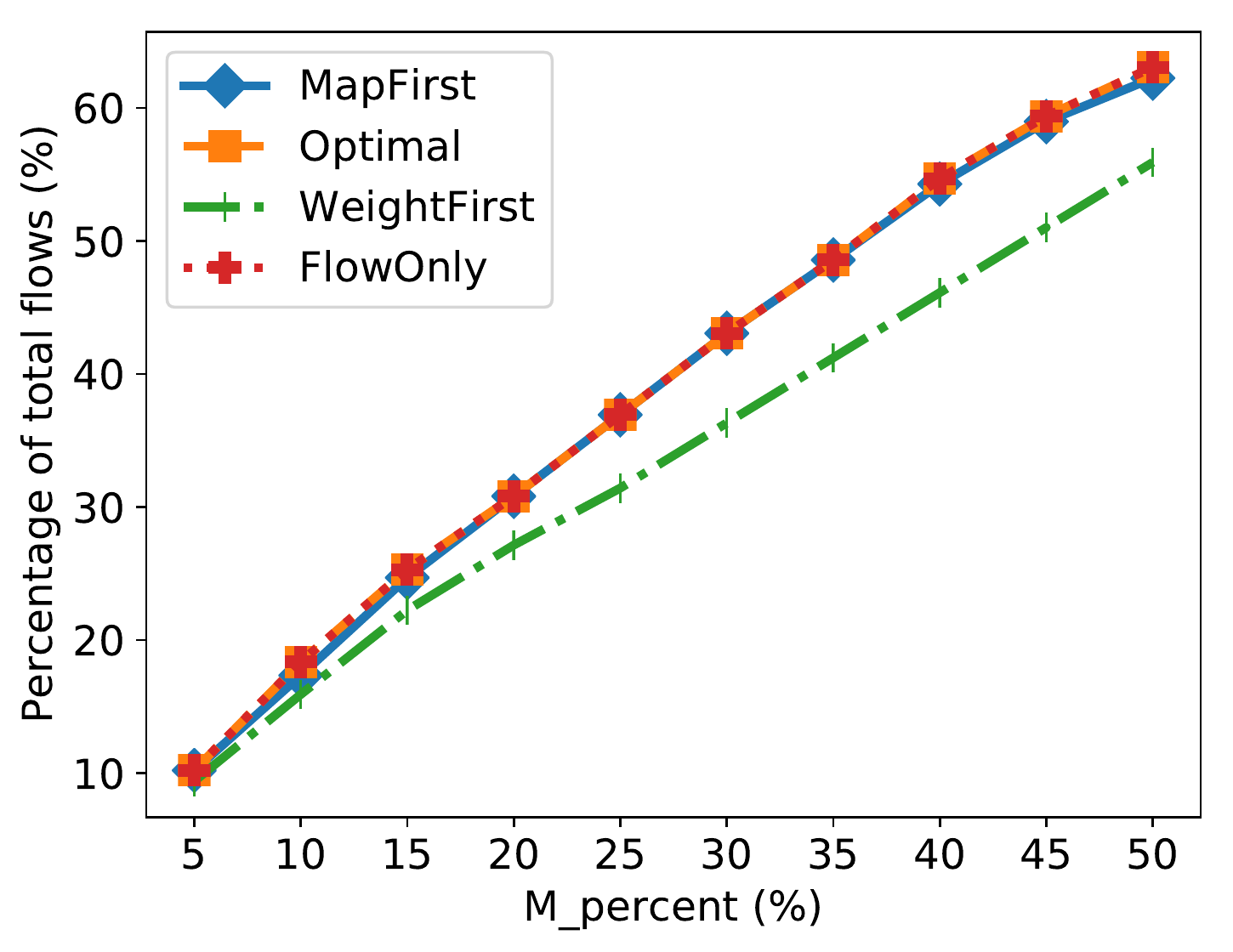}
}
\caption{Number of programmable flows. The higher, the better.}
\label{fig:flow}
\end{figure*}

Fig. \ref{fig:flow} shows the performance of the programmable flows of different algorithms. In all the three topologies, the performance of the four algorithms increases as $M\_percent$ becomes larger. Optimal and FlowOnly are the optimal solutions to maximize the number of programmable flows. We can observe from Fig. \ref{fig:flow}  that WeighFirst's performance is the worst and \solution's performance is very close to that of Optimal and FlowOnly. Recall weight $ \omega_{ij}$ is equal to $R_i - \lambda D_{ij} $ in WeightFirst. The maximum value of $ \omega_{ij}$ is $ R_i$ when deploying a controller at the location of a selected switch. Since $\lambda$ is usually a small value, $R_i$ plays the dominant role in $\omega_{ij}$ and thus WeightFirst first greedily tests the switch-controller pair based on the descending order of the number of flows in switches. Even though WeightFirst considers the delay in the later tests, it only focuses on a single switch-controller pair without the global view of the problem. In sharp contrast, \solution \ also tests the switch-controller pair but based on the result of the LP relaxation. The LP relaxation considers the entire problem to generate its result that reflects the probability of selecting switch-controller pairs. Therefore, the testing order of switch-controller pairs in \solution \ is more efficient than WeightFirst. This is the reason why \solution \ achieves much better performance than WeightFirst. 

\begin{figure*}[t]
\centering
\subfigure[\topoa]{
\includegraphics[width=2in]{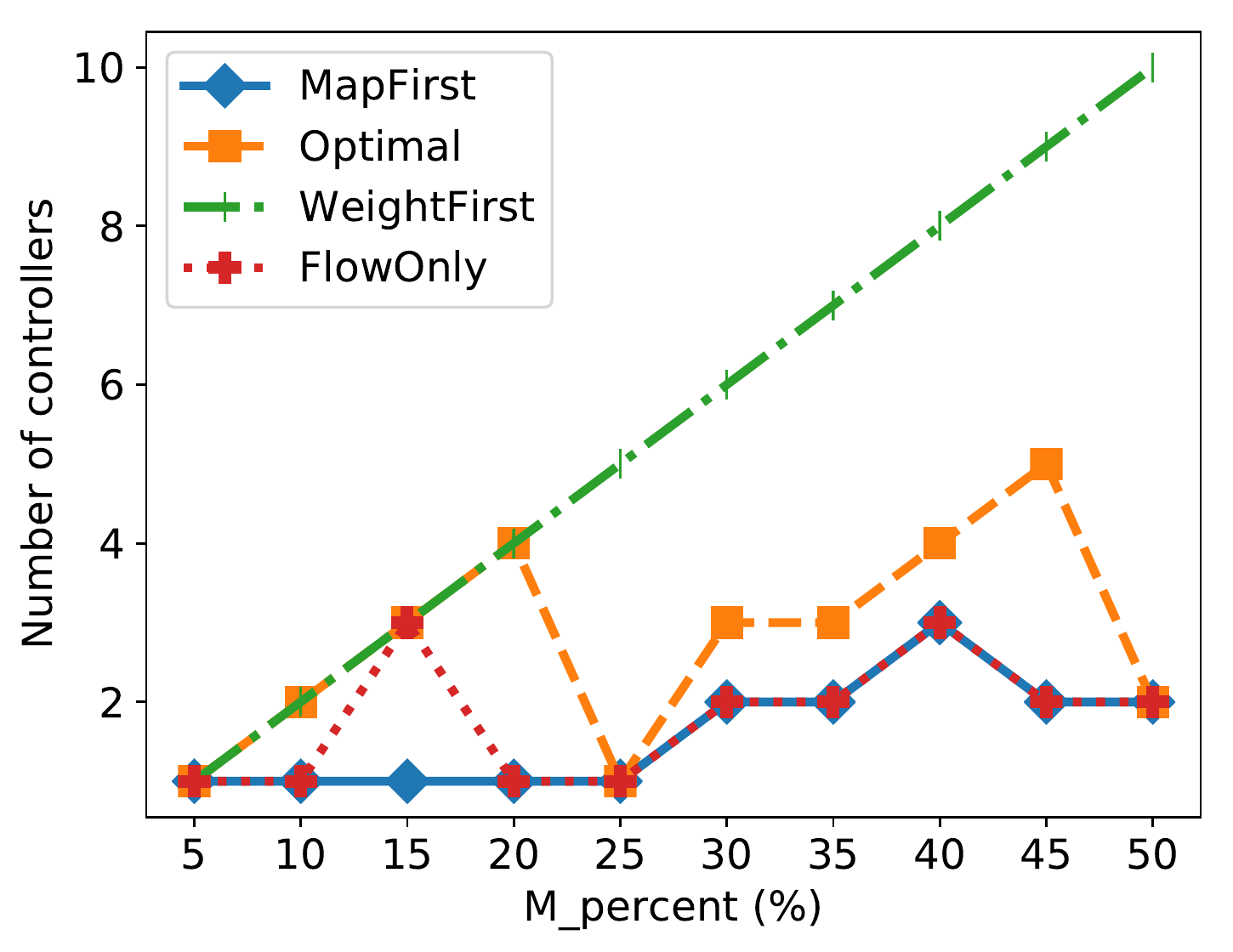}
}
\subfigure[\topob]{
\includegraphics[width=2in]{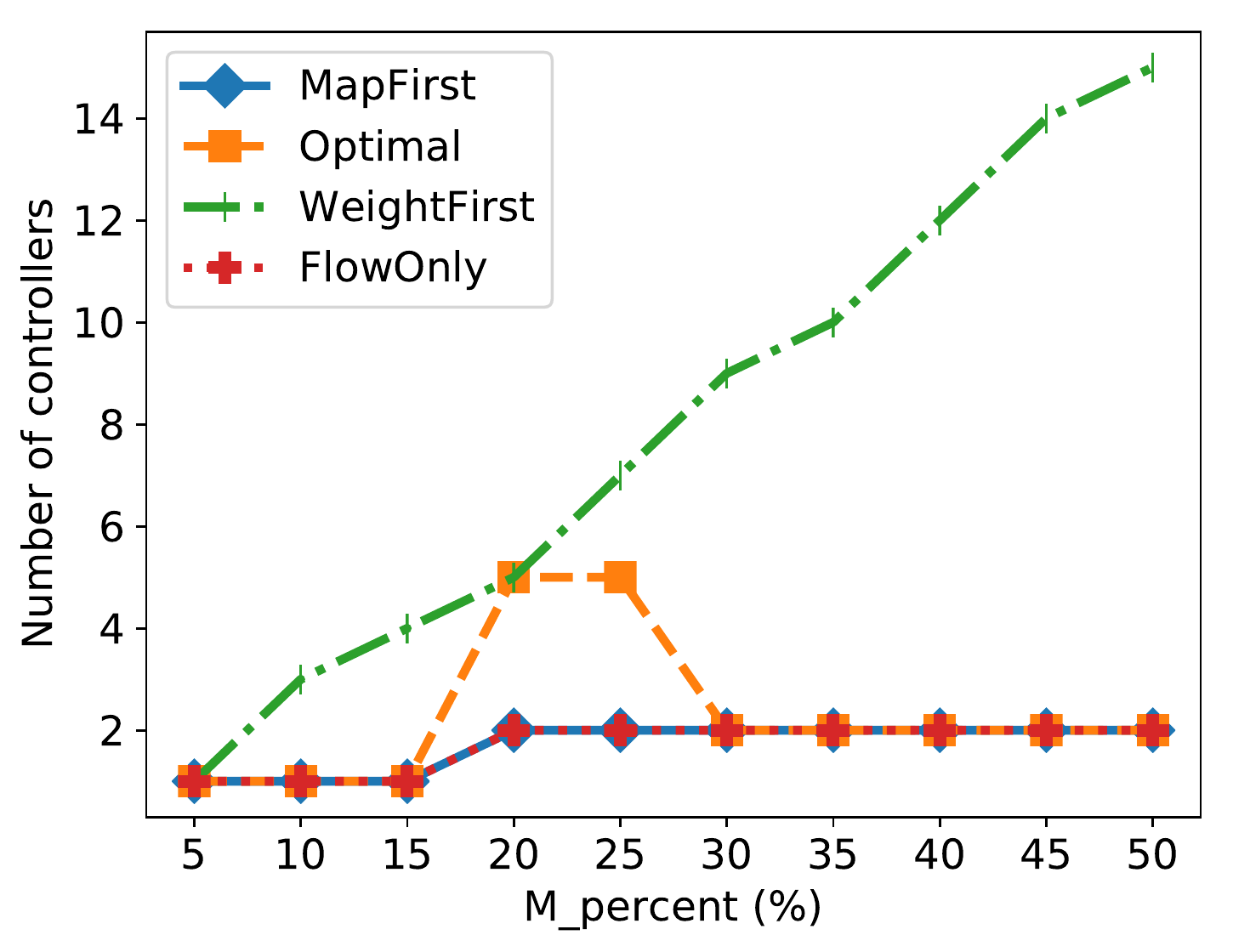}
}
\subfigure[\topoe]{
\includegraphics[width=2in]{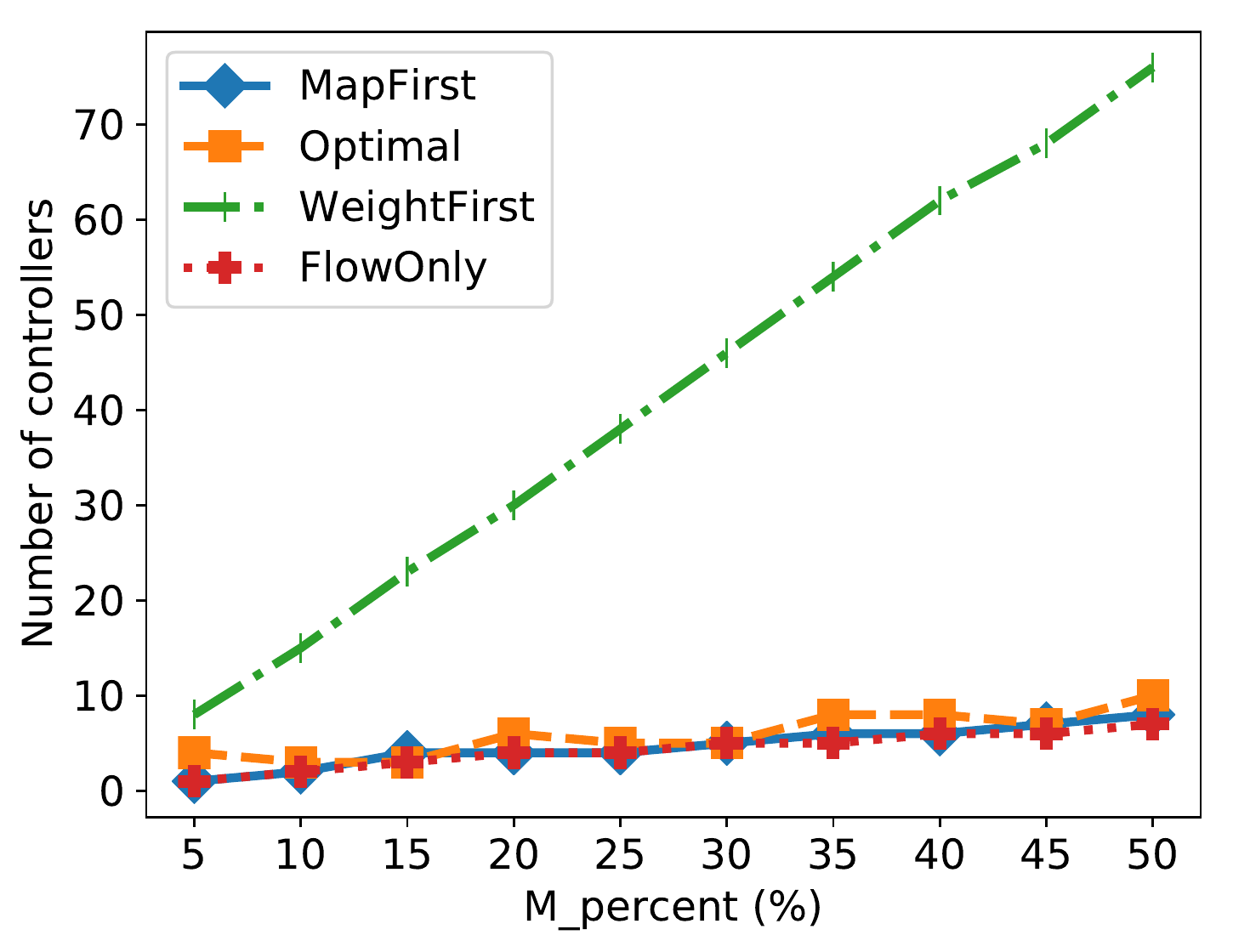}
}
\caption{Number of controllers. The lower, the better.}
\label{fig:con}
\end{figure*}

\subsection{Performance of deployed controllers}

Fig. \ref{fig:con} shows the number of controllers deployed by different algorithms under the three topologies. WeightFirst deploys more controllers than all the other algorithms since it always deploys a controller for each upgraded switch. Thus, under the same upgrade budget, WeightFirst upgrades fewer switches than all the other algorithms. FlowOnly performs the best since it does not consider the propagation delay in its objective function and deploys enough controllers for the upgraded SDN switches. Optimal deploys more switches than FlowOnly and \solution. Recall the cost ratio of an SDN switch and a controller is $\gamma$. If the upgrade budget is not enough to upgrade a switch, the rest budget can also be used to deploy 1 to $\gamma-1$ controllers. In this case, \solution \ and Optimal make different decisions: \solution \ does not deploy extra controllers and just stop the program (see lines \ref{ln:greedy_cost_constraint}-\ref{ln:greedy_cost_constraint_end} in Algorithm \ref{alg:lp}), while Optimal will deploy more controllers to minimize the propagation delay between SDN switches and controllers. 

\subsection{Performance of the propagation delay}
Fig. \ref{fig:delay} shows the propagation delay between SDN switches and controllers. In all topologies, WeightFirst performs the best as it deploys one controller at each SDN switch in most cases. Among all algorithms, FlowOnly performs the worst because problem formulation \eqref{stage1-a} does not consider the propagation delay between SDN switches and controllers. In all the three topologies, Optimal and \solution's propagation delays increase as $M\_percent$ increases, but their increasing rate is much lower than FlowOnly's. Compared to \solution, Optimal achieves a better propagation delay performance since it deploys more controllers to minimize the propagation delay. From Fig. \ref{fig:con} and \ref{fig:delay}, we can see that in most cases, when \solution \ and Optimal deploy the same number of controllers, they have the same performance, except two cases $M\_percent = 50\%$ in Fig. \ref{fig:con}(a) and $M\_percent = 15\%$ in Fig. \ref{fig:con}(b). In these two cases, Optimal performs better than \solution \ even though they deploy the same number of controllers. 

\begin{figure*}[t]
\centering
\subfigure[\topoa]{
\includegraphics[width=2in]{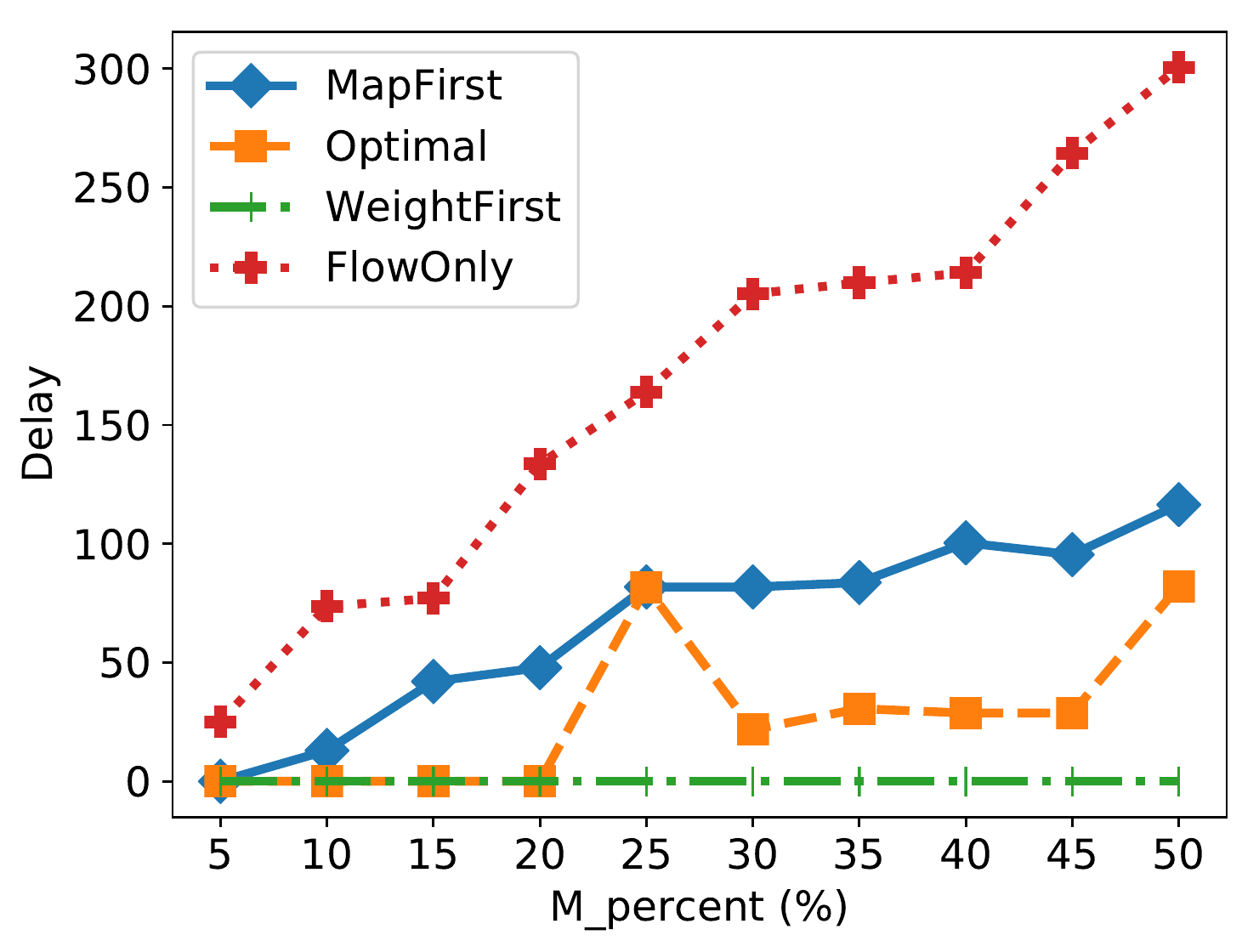}
}
\subfigure[\topob]{
\includegraphics[width=2in]{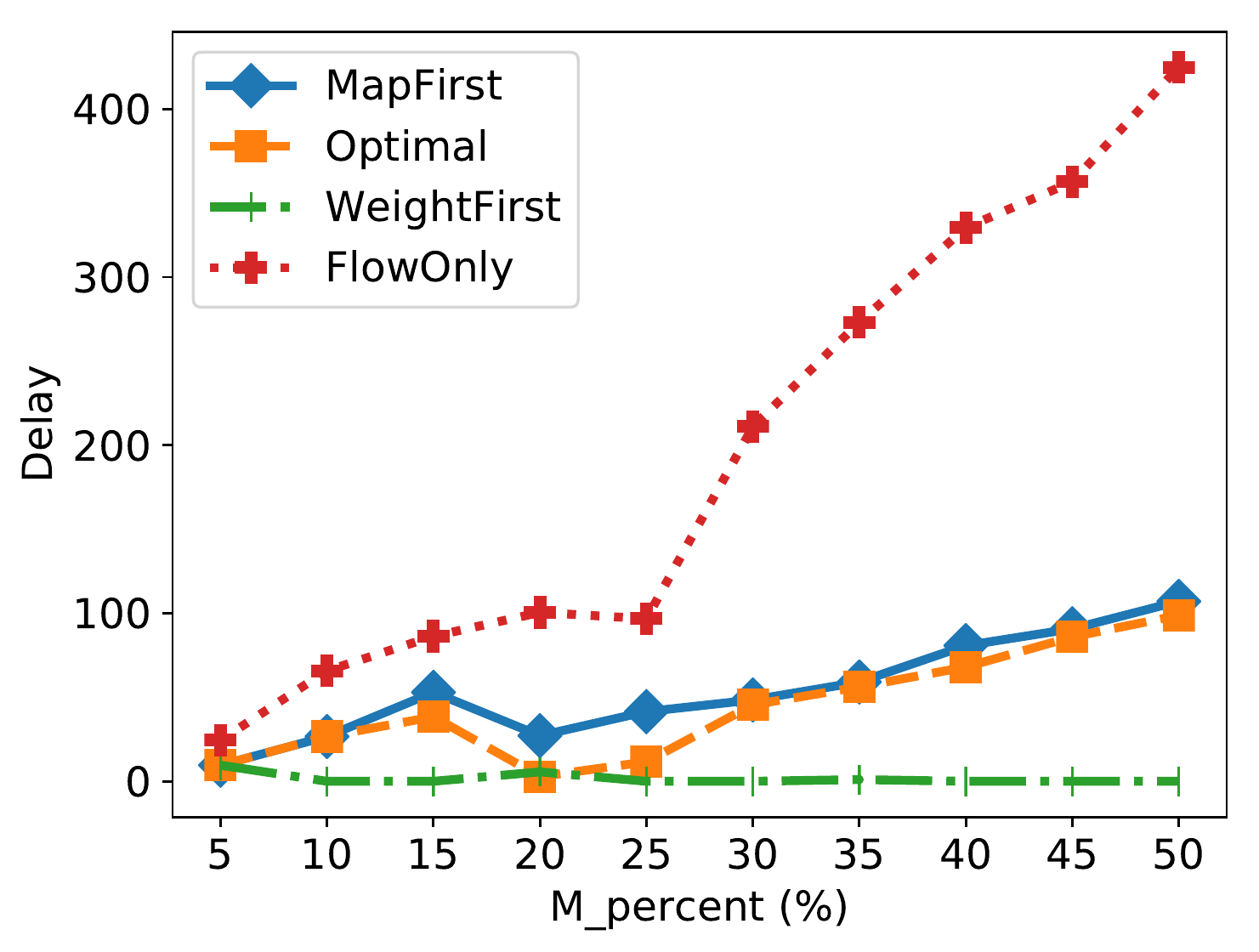}
}
\subfigure[\topoe]{
\includegraphics[width=2in]{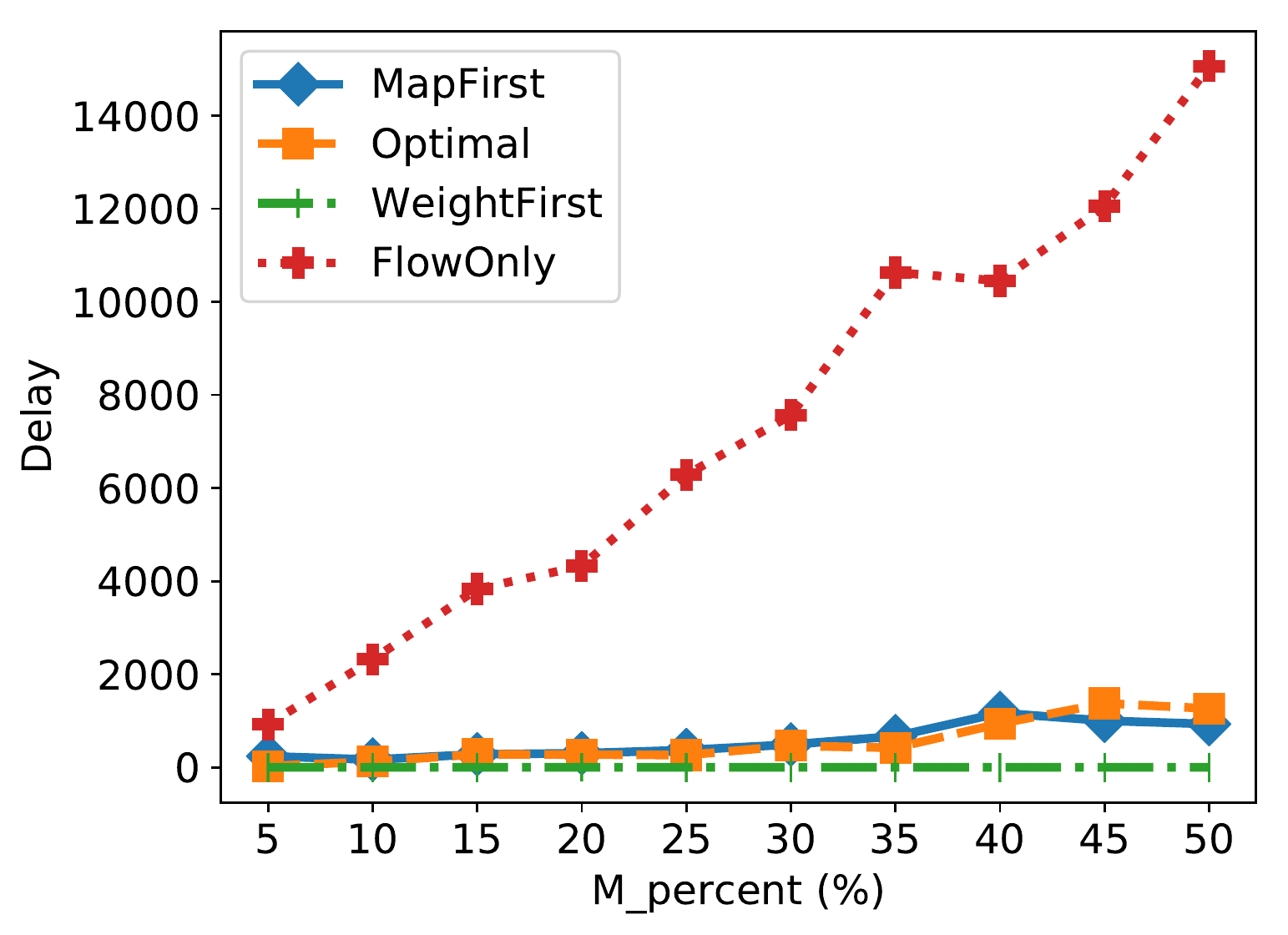}
}
\caption{Propagation delay between SDN switches and controllers. The lower, the better.}
\label{fig:delay}
\end{figure*}

\begin{figure*}[t]
\centering
\subfigure[\topoa]{
\includegraphics[width=2in]{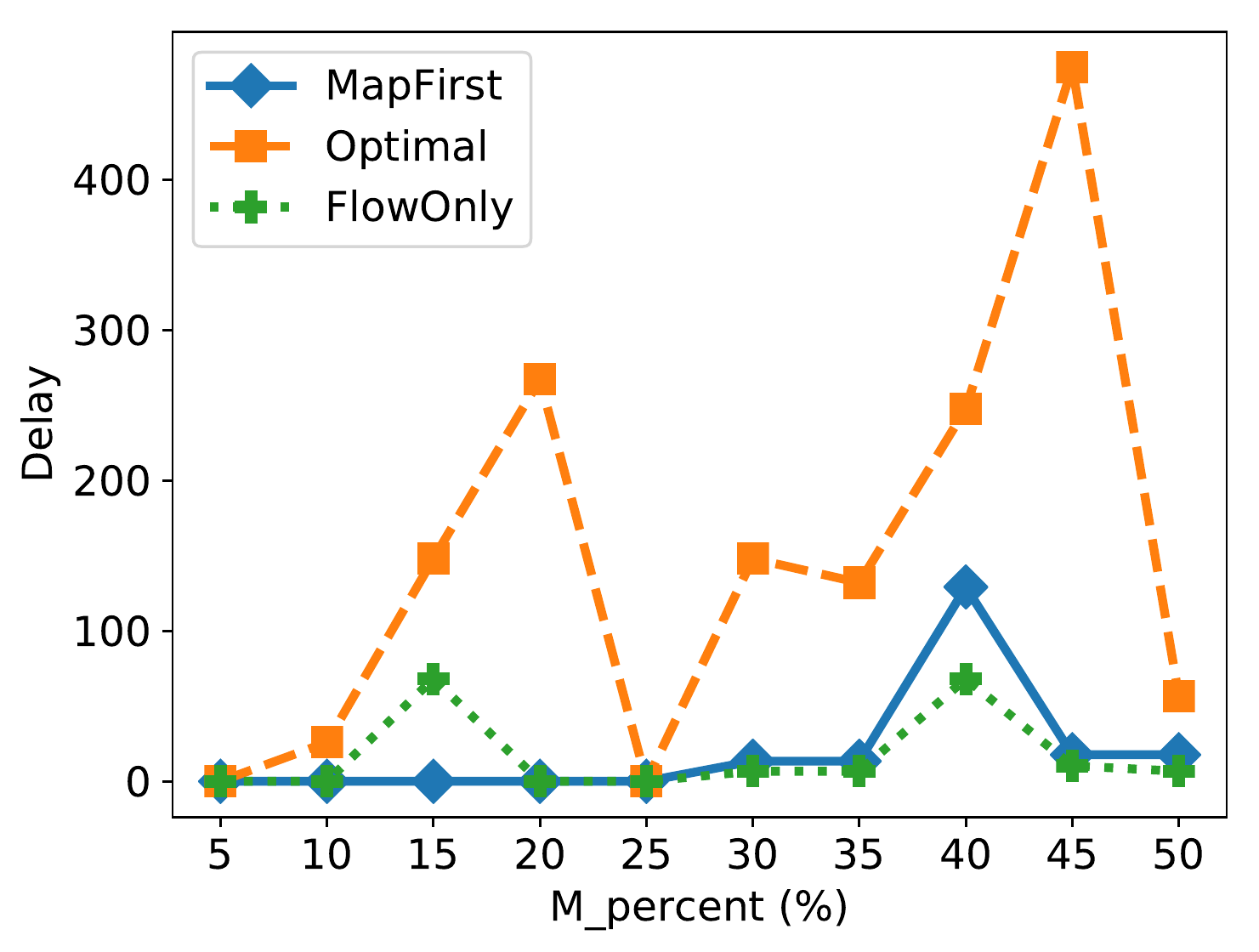}
}
\subfigure[\topob]{
\includegraphics[width=2in]{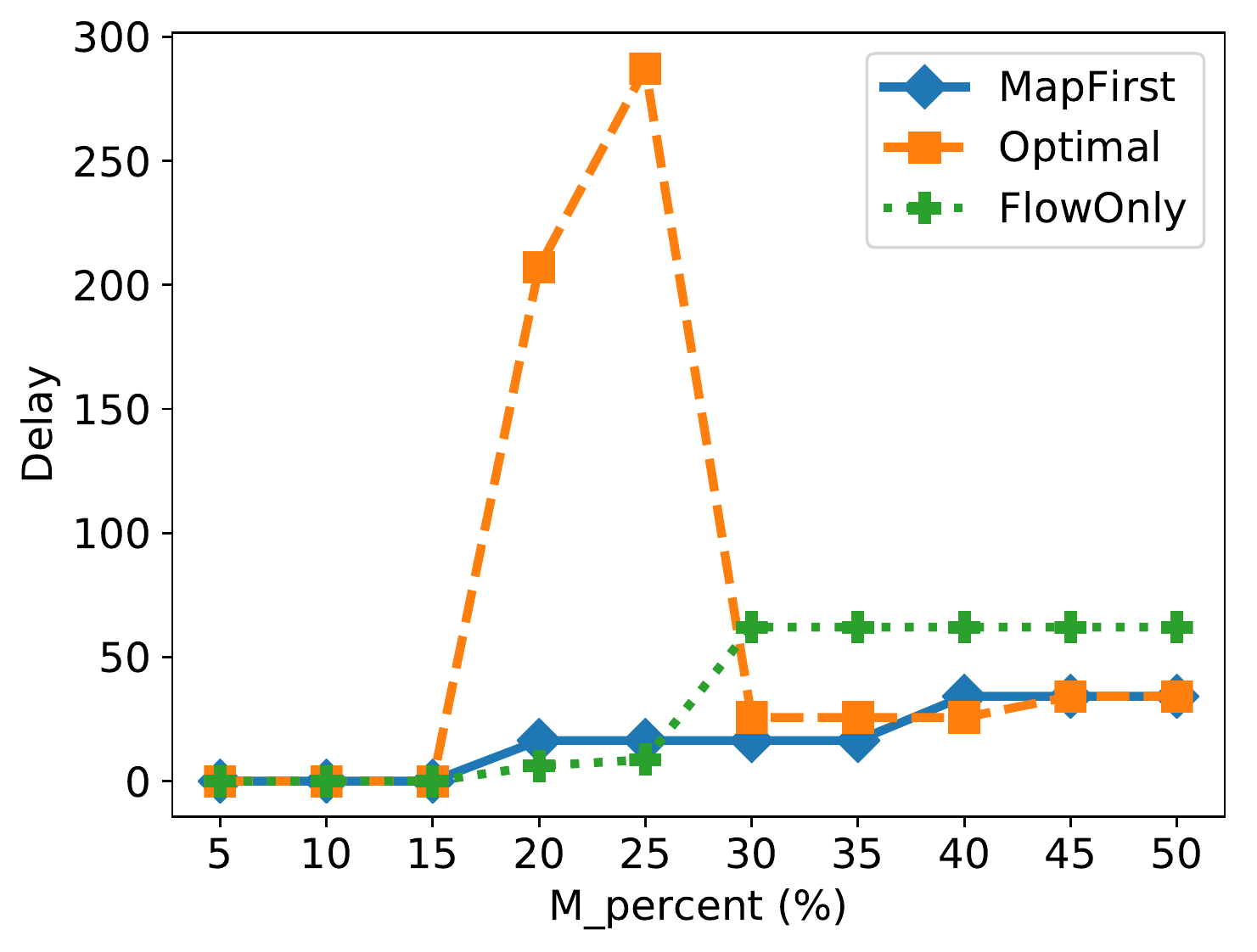}
}
\subfigure[\topoe]{
\includegraphics[width=2in]{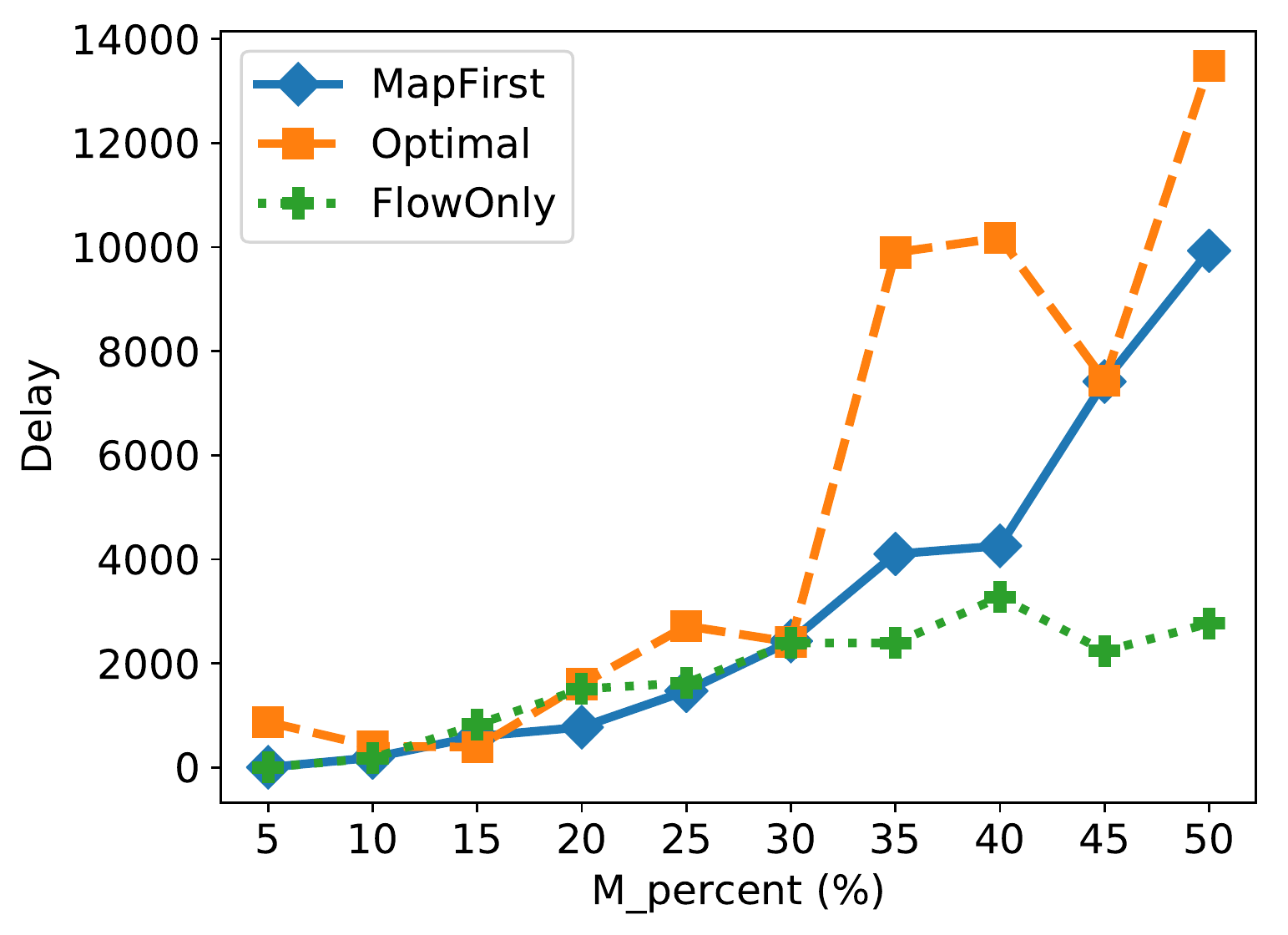}
}
\caption{Propagation delay between controllers. The lower, the better.}
\label{fig:con-delay}
\end{figure*}

Fig. \ref{fig:con-delay} shows the propagation delay between controllers. We do not show the results of WeightFirst because its performance is very bad. We can observe, from the figure, Optimal performs the worst since it does not consider this factor in the problem formulation (An interesting future is to consider the propagation delay in the problem formulation). FlowOnly performs better than Optimal since it deploys much fewer controllers than Optimal (see Fig. \ref{fig:con}), and the deployed controllers are near to each other. In most cases, \solution \ has the best performance. This is because \solution \ always deploys the minimum number of controllers. In our problem formulation, we use the budget constraint to implicitly limit the number of controllers. In particular, our objective is to maximize the number of flows, which is equivalent to maximize the number of upgraded SDN switches. Given the upgrade budget, upgrading more switches will reduce the number of controllers to deploy.

\begin{figure*}[t]
\centering
\subfigure[\topoa]{
\includegraphics[width=2in]{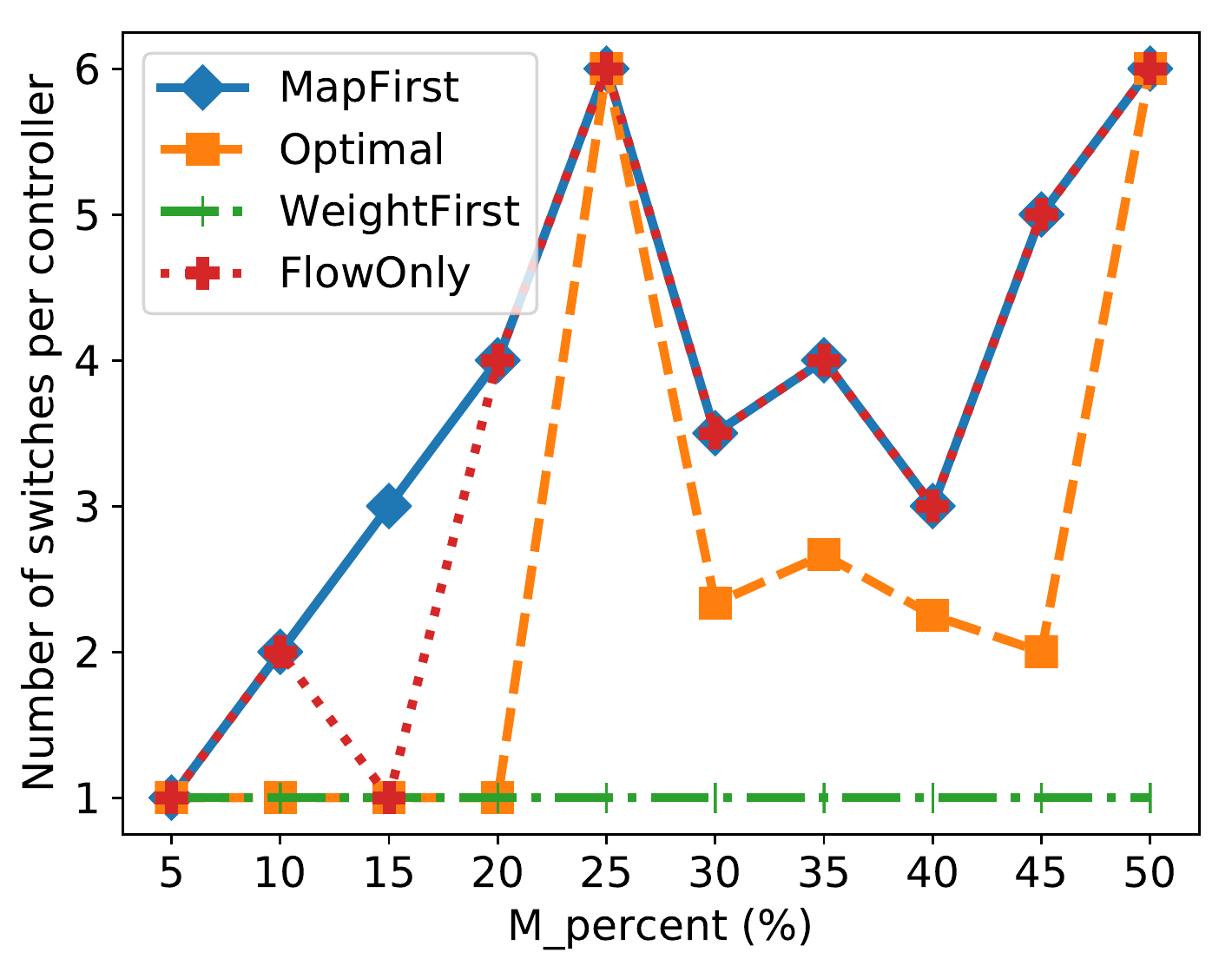}
}
\subfigure[\topob]{
\includegraphics[width=2in]{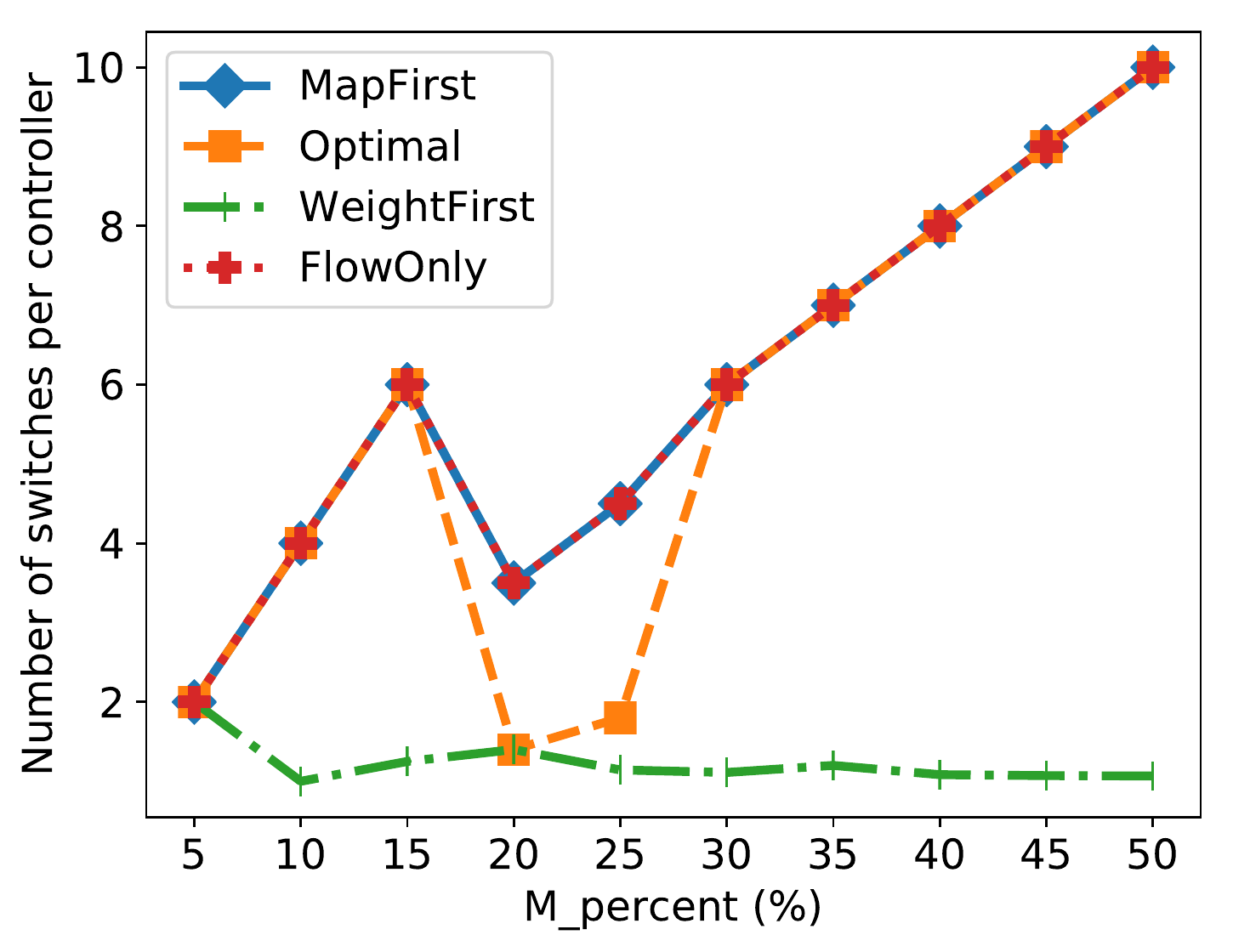}
}
\subfigure[\topoe]{
\includegraphics[width=2in]{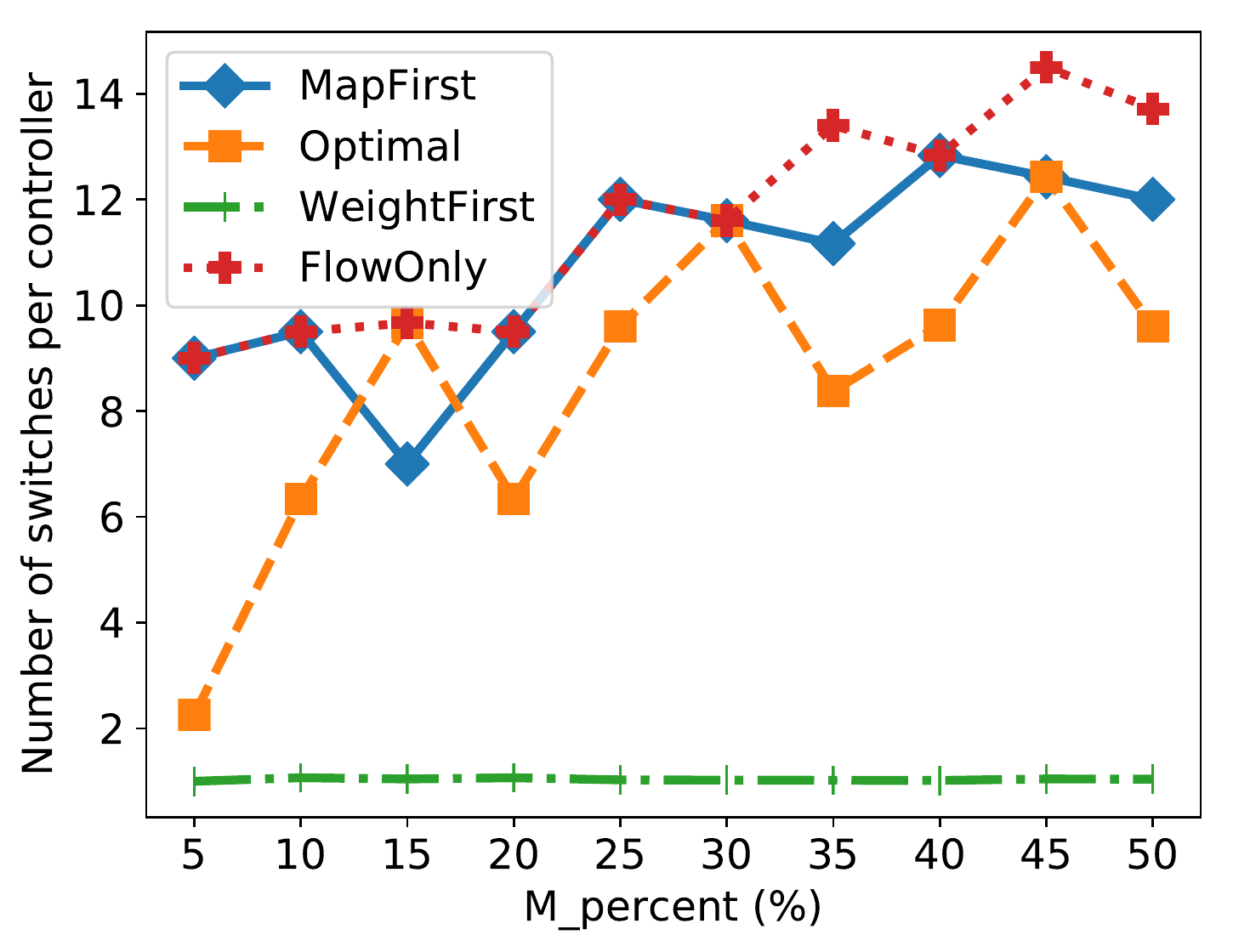}
}
\caption{Ratio of the number of upgraded switches to the number of deployed controllers. The higher, the better.}
\label{fig:ratio}
\end{figure*}

\subsection{Performance of controller control ability on switches}
We use the ratio of the number of switches to the number of controllers to measure the controller control ability on switches. If a controller can control many switches, it will simplify the network control. Fig. \ref{fig:ratio} shows the controller control ability performance of all algorithms. In most cases, \solution \ performs the best since it upgrades the same number of switches as Optimal but deploys fewer controllers.

\subsection{Summary of simulation findings}
From our simulation results, we can conclude that the following two design principles of \solution \ are crucial: (1) among the three variables (i.e., the selection of upgraded switches, the selection of deployed controllers, and the mappings between SDN switches and controllers), the mapping variables are more important than the other two variables since they reflect the interrelationship between the other two variables, and (2) the solution of the LP relaxation can effectively reveal the importance of the mappings and thus the solution structure of the problem. WeightFirst only takes the first principle into consideration and performs the worst. Optimal performs the best in terms of the two objectives: maximizing the number of programmable flows and minimizing the propagation delay between SDN switches and controllers. However, it will deploy more controllers and has a higher propagation delay among controllers than \solution. More importantly, its worst-case complexity is significantly larger than \solution \ and WeightFirst. In summary, compared to Optimal, \solution \ not only achieves a comparable performance of the two objectives with significantly lower complexity but also deploys less number of controllers and thus has a shorter propagation delay among controllers. Considering all performance metrics, \solution \ achieves the best performance. Therefore, to efficiently handle a similar problem with interrelated variables, one can use our design principles.

\section{Related works}
\label{relatedwork}
The evolution from legacy networks to SDN is a long journey \cite{kreutz2015software}. Considering SDN still being a fast developing technology, the full deployment of SDN not only requires a huge upgrade cost but also brings unexpected risks. The hybrid SDN is a promising alternative option and is receiving quickly increasing attention from industry and academia. Vissicchio et al. \cite{vissicchio2014opportunities} first analyzed different models of hybrid SDNs. The research of the hybrid SDN can be categorized into two classes based on the application scenarios:

\subsection{Layer 2 hybrid SDN}
Panopticon \cite{levin2014panopticon} proposed an optimization framework to determine the partial SDN deployment and assumed that each flow in the network traverses at least one SDN switch. HybNET \cite{lu2013hybnet} designed a configuration mechanism to automatically translate network configurations between legacy networks and SDN networks. Telekinesis \cite{jin2015telekinesis} manipulated forwarding entries in the legacy switches with a customized flow control primitive and enabled the forwarding flow on a user-defined path rather than the path in the spanning tree. Magneto \cite{jin2017magneto} improved the work in \cite{jin2015telekinesis} by providing more fine-grained flow controls and quick and stable user-defined path establishment.

\subsection{Layer 3 hybrid SDN}
\subsubsection{Switch upgrade}
Poularakis et al. \cite{poularakis2017one} and Jia et al. \cite{jia2016incremental} proposed to incrementally upgrade  SDN to maximize the amount of programmable traffic under the given upgrade budget constraint. Xu et al. \cite{xu2017incremental} compared the performance of hybrid SDNs by either replacing legacy devices with SDN devices or adding new SDN devices.  Caria et al. \cite{caria2015divide} exploited the property of the network topologies and considered the centrality in the network to update switches in a hybrid SDN. 

\subsubsection{Traffic engineering}
Agarwal et al. \cite{agarwal2013traffic} first formulated an optimization problem for achieving traffic engineering in SDN with the partial deployment of SDN devices. Fibbing \cite{vissicchio2015central} introduced a centralized control over distributed IP routing by injecting crafted routing messages via OSPF and enhanced the flexibility, diversity, and reliability of L3 routing. Chu et al. \cite{chu2015congestion} designed an approach to fast recover from single link failure while maintaining load-balancing performance for the post-recovery network. Wang et al. \cite{wang2016saving} and Jia et al. \cite{jia2018intelligent} explored power saving in Layer 3 hybrid SDN. Guo et al. \cite{guo2014traffic} proposed to adjust the weights of links and flow split ratio at the SDN nodes to achieve load balancing in a given hybrid SDN. Hong et al. \cite{hong2016incremental} proposed to satisfy a variety of traffic engineering goals in the hybrid SDN. However, all the above works did not introduce the real deployment of the SDN control plane and did not consider the impact of the control plane deployment on the hybrid SDN.

\subsection{Multi-controller data plane}
\subsubsection{Controller deployment}
Controller deployment is an important issue in multi-controller research. In WANs, the propagation latency is the critical part of the total latency \cite{heller2012controller}\cite{yao2014capacitated}, and some works aim to minimize the propagation latency among controllers and switches \cite{hu2018multi}\cite{hu2018bidirectional}. In data center networks, Wang et al. \cite{wang2016dynamic} proposed to dynamically map switches to controllers to mitigate the load imbalance among controllers and reduce the response time. Wang et al. \cite{wang2017efficient} considered the maintenance cost of the controller cluster and assigned controllers to minimize the total cost of controller response time and maintenance on the cluster of controllers.

\subsubsection{Resiliency}
Resiliency is a critical concern for designing the multi-controller data plane. He et al. \cite{li2014byzantine} presented to manage each SDN device with multiple controllers at the cloud to resist Byzantine attacks on controllers and the communication links between controllers and SDN switches. Hu et al. \cite{hu2018adaptive} \cite{hu2019dynamic} designed a fault-tolerant multi-controller control plane that mitigates the performance degradation of the controller chain failure caused by unreasonable slave controller assignment. MORPH \cite{MORPH} is a multi-controller framework, which is tolerant to unavailability failures of SDN controllers and Byzantine failures caused by malicious attacks by efficiently distinguishing and localizing faulty controller instances and appropriately reconfiguring the control plane. In the future, we will consider this factor to extend the work in this paper.

\section{Conclusion and future work}
\label{conclusion}
In this paper, we considered the impact of the control plane on the hybrid SDN and proposed to deploy the multi-controller control plane during upgrading a legacy network to a hybrid SDN. We formulated an optimization problem that jointly maximizes the number of flows from SDN switches and minimizes the propagation delay of flow requests between the SDN's control plane and data plane under the given upgrade budget constraint. By carefully analyzing the problem's structure, we proposed efficient solutions to solve the problem. The simulation results based on the real network topologies show the effectiveness and efficiency of our solutions. In the future, we will consider other practical factors in our problem formulation, such as the variation of the traffic pattern, the queuing delay in controllers, and multiple network upgrades.

\bibliographystyle{IEEEtran}
\bibliography{hsdn}
\end{document}